\documentclass[11pt, letter]{article}

\usepackage{stolenstyle}

\usepackage{amsmath,epsf,amssymb,latexsym,amsthm,setspace,bbm,array,pifont,ulem,enumerate}

\usepackage[matrix,arrow,frame,import,curve,color]{xy}
\xyoption{arc}

\DeclareFontFamily{U}{rsf}{}
\DeclareFontShape{U}{rsf}{m}{n}{
  <5> <6> rsfs5 <7> <8> <9> rsfs7 <10-> rsfs10}{}
\DeclareMathAlphabet\Scr{U}{rsf}{m}{n}

\def\iden{{\mathbbm 1}}

\def\cDb{{\overline{\cD}}}

\def\GUL{\GU(1)_{\text{L}}}
\def\GUR{\GU(1)_{\text{R}}}

\def\C{{\mathbb C}}

\def\P{{\mathbb P}}
\def\R{{\mathbb R}}
\def\Z{{\mathbb Z}}

\def\End{\operatorname{End}}

\def\Hom{\operatorname{Hom}}

\def\Pic{\operatorname{Pic}}

\def\Vol{\operatorname{Vol}}

\def\deg{\operatorname{deg}}

\def\rank{\operatorname{rank}}

\def\GL{\operatorname{GL}}

\def\GU{\operatorname{U{}}}

\def\p{\partial}

\def\la{\langle}
\def\ra{\rangle}

\def\cA{{\cal A}}
\def\cB{{\cal B}}
\def\cC{{\cal C}}
\def\cD{{\cal D}}
\def\cE{{\cal E}}

\def\cJ{{\cal J}}

\def\cL{{\cal L}}
\def\cM{{\cal M}}

\def\cO{{\cal O}}
\def\cP{{\cal P}}

\def\cS{{\cal S}}

\def\cW{{\cal W}}

\newcommand\rhot{\widetilde{\rho}}
\newcommand\sigmat{\widetilde{\sigma}}

\newcommand\vphi{\varphi}

\newcommand\Gammab{\overline{\Gamma}}

\newcommand\Gammat{\widetilde{\Gamma}}

\newcommand\Sigmat{\widetilde{\Sigma}}

\newcommand\Eh{\widehat{E}}

\newcommand\Db{\overline{D}}

\newcommand\Vb{\overline{V}}

\newcommand\Jt{\widetilde{J}}

\newcommand\Mt{\widetilde{M}}
\newcommand\Nt{\widetilde{N}}

\newcommand\Xt{\widetilde{X}}

\def\inter{\operatorname{int}}
\def\Gr{\operatorname{Gr}}

\def\br{{\boldsymbol{r}}}
\def\bE{{\boldsymbol{E}}}

\def\cDb{{{\overline{\cD}}}}
\def\cMb{{{\overline{\cM}}}}

\def\mon{{{\mathsf{M}}}}

\def\F{{\mathbb F}}

\def\nSUSY{{{\text{\sout{SUSY}}}}}

\def\relint{\operatorname{relint}}
\def\MV{\operatorname{MV}}
\def\Conv{\operatorname{Conv}}

\def\tsing{{\text{sing}}}
\def\preprint{ {{~}}}

\newtheorem{thm}{Theorem}[section]
\newtheorem{lem}[thm]{Lemma}

\title{Global aspects of (0,2) moduli space: toric varieties and tangent bundles}
\author[a] {Ron Donagi,}
\author[b] {Zhentao Lu,}
\author[c] {and Ilarion V.~Melnikov}
\affiliation[a]{Department of Mathematics \\
University of Pennsylvania, Philadelphia, PA 19104, USA}
\affiliation[b]{Mathematical Institute \\ University of Oxford, Oxford OX2 6GG, UK}
\affiliation[c]{Department of Mathematics \\ Harvard University, Cambridge, MA 02138, USA}
\emailAdd{donagi@math.upenn.edu}
\emailAdd{Zhentao.Lu@maths.ox.ac.uk}
\emailAdd{ilarion@math.harvard.edu}

\abstract{We study the moduli space of A/2 half-twisted gauged linear sigma models for NEF Fano toric varieties.  Focusing on toric deformations of the tangent bundle, we describe the vacuum structure of many (0,2) theories, in particular identifying loci in parameter space with spontaneous supersymmetry breaking or divergent ground ring correlators.  We find that the parameter space of such an A/2 theory and its ground ring is in general a moduli stack, and we show in examples that with suitable stability conditions it is possible to obtain a simple compactification of the moduli space of smooth A/2 theories.}

\begin{document}

\maketitle

\section{Introduction}\label{s:intro}
Over its nearly thirty-year history the heterotic string has played an important role in string theory and mathematical physics for two key reasons: it plays a fundamental role in the web of string dualities and leads to a vast landscape of string-perturbative N=1 supersymmetric d=4 vacua, where the internal degrees of freedom are described by a d=2 (0,2) superconformal field theory (SCFT), and the low-energy spacetime physics is a chiral gauge theory coupled to supergravity.  A typical construction of this sort has a large number of marginal deformations, and in string perturbation theory
many of these correspond to exactly flat directions.  The geometry of the resulting moduli space contains a great deal
of information about the SCFT and the spacetime physics of the compactification, but even in very concrete models it is
difficult to describe.  It is a sobering thought that to date there are no good methods to compute the moduli space metric in
genuine (0,2) theories!  Present techniques do allow us to examine the moduli space as a complex variety, 
to characterize singular loci, provide natural compactifications and to relate different (perhaps dual) descriptions of the same SCFT.  For example, such a description has played a role in recent developments in (0,2) mirror symmetry~\cite{Melnikov:2012hk}.

The goal of this work is to begin a systematic study of these (0,2) moduli spaces in the context of (0,2) gauged linear sigma models (GLSMs)~\cite{Witten:1993yc,Distler:1993mk}, paying special attention to global structures and interpretations in algebraic geometry.  Various aspects of these theories have been examined before.  For instance, local properties of the deformation space and (0,2) resolutions of singularities were considered in~\cite{Distler:1996tj,Chiang:1997kt}, and many efforts were focused on world-sheet instanton obstructions to classically marginal deformations (for instance~\cite{Silverstein:1995re,Berglund:1995yu,Basu:2003bq,Beasley:2003fx}).  

In this paper we consider a model problem:  the torus-equivariant (or ``toric'' for short) deformations of the tangent bundle of a smooth Fano toric variety $X$.  Such a bundle, together with a choice of complexified K\"ahler class on $X$,  defines a (0,2) GLSM and an associated A/2 quasi-topological theory~\cite{McOrist:2008ji}.\footnote{There are no complex structure deformations since smooth toric Fano varieties are rigid.  This statement and various generalizations are reviewed in~\cite{MR2931867}.} Our goals are: (i) provide a global description of the deformation space that unifies the bundle and K\"ahler data; (ii) identify singular loci; and (iii) describe possible compactifications of the space of smooth (0,2) theories.
Our global approach follows from the point of view developed in~\cite{Kreuzer:2010ph}, itself a natural generalization of the presentation of the ``polynomial'' deformations of complex structure for (2,2) theories based on Calabi-Yau hypersurfaces developed in~\cite{Aspinwall:1993rj} and reviewed in~\cite{Cox:2000vi}.  Along the way we obtained a number of results in toric geometry and combinatorics that we will quote below, leaving details of the proofs to the companion paper~\cite{DonagiLuIVM2}~.

A GLSM for a Fano toric variety does not lead to a (0,2) SCFT: for generic choice of bundle deformation parameters the theory flows to a gapped (0,2) theory, while for special degenerate values the result is rather a gapless theory with spontaneously broken supersymmetry.\footnote{This is not in conflict with Witten index/elliptic genus computations precisely because at these degenerate values the asymptotics of the scalar potential are modified~\cite{Tong:2008qd}; see also~\cite{Gadde:2013lxa} for a recent discussion of spontaneous SUSY breaking in the context of non-abelian GLSMs.  Note that these theories all have a non-anomalous vectorial R-symmetry, which forbids deformations of the SUSY current algebra recently discussed in~\cite{Dumitrescu:2011iu}.}  Nevertheless, there are good reasons to study these theories.  First, experience with GLSMs has shown that results for gapped models associated to compact projective toric varieties often carry over to GLSMs for Calabi-Yau geometries; this is true both in the (2,2) ~\cite{Witten:1993yc,Morrison:1994fr} and (0,2)~\cite{McOrist:2007kp,McOrist:2008ji} cases.  
Second, the massive theories are interesting in their own right:  while one complexified K\"ahler parameter is transmuted into a scale, the theories still have a rich vacuum structure that depends on the remaining parameters.  This is of interest for physical applications, where these massive theories can describe low energy excitations of solitonic strings and a wide variety of surface defects in four-dimensional gauge theories (see, e.g.~\cite{Tong:2008qd,Gaiotto:2013sma} and references therein).  Moreover, the topological heterotic ring~\cite{Adams:2003zy,Katz:2004nn,Adams:2005tc}, the ground ring of the quasi-topological A/2 model~\cite{McOrist:2008ji}, leads to the theory of quantum sheaf cohomology, a generalization of quantum cohomology, that was computed for tangent bundle deformations of  compact toric varieties in~\cite{McOrist:2007kp} and given a rigorous mathematical framework in~\cite{Donagi:2011uz,Donagi:2011va}.  Finally, and most pragmatically, this class of models, while sharing many features with the more elaborate GLSMs that flow to non-trivial SCFTs, has a deformation space that, on one hand can be given a very general description, and on the other hand is readily studied in concrete examples.  

Before plunging into the details we end this introduction with a summary of our main results on (0,2) torus-equivariant deformations of (2,2) GLSMs associated to a smooth Fano toric variety $X$.

Our first statement concerns the classical algebraic geometry of these deformations, which also describes the classical large radius limit of the GLSM.  Toric vector bundles have a complete but complicated description due to Klyachko~\cite{Klyachko:1989eb}: see~\cite{Knutson:1997yt,Payne:2008tb} for a modern exposition and references therein for a history of the subject.  We show that all toric deformations of the tangent bundle can be realized by deforming the Euler short exact sequence that defines the tangent sheaf --- physically, this means that every such deformation can be realized in the GLSM.  The relationship between our presentation and that of~\cite{Payne:2008tb} is rather like that of the homogeneous coordinate ring of a toric variety to the affine patch data:  while some features are much simpler to understand from the latter point of view, many global features are more easily grasped from the former perspective.

Next we consider the quantum theory, where in addition to the bundle deformations we must take account of the complexified K\"ahler parameters.  The important lesson is that the two must be considered together; an invariant distinction between the two classes of data only arises upon taking certain classical limits (analogous to a large radius limit in Calabi-Yau compactification).  The result is a moduli stack, presented as a quotient $\{\Gr(k,k+d) \times T_k\} / T_{k+d}$, where $d = \dim X$,  $k = h^{1,1}(X)$, and $T_n = (\C^\ast)^n$.  A generic point in this stack leads to a smooth theory, but there is a $T_{k+d}$--invariant locus $\cA_{\tsing} \subset \Gr(k,k+d) \times T_k$ where the theory is singular --- i.e. the A/2 correlators diverge.  The moduli stack is not a separated variety, so that its geometric properties are a bit obscure.   Fortunately, as we show in examples, it is possible to choose subsets $F \subset \cA_\tsing$ such that $\cMb(X) = \{ \Gr(k,k+d) \times T_k \setminus F\}/T_{k+d}$ is a separated variety that is a compactification of the space of smooth theories $\cM(X)$.   Interestingly, for every choice of $F \subset D$ the singular points in $\cMb(X)$ include parameter values where SUSY is spontaneously broken.

The rest of the paper is organized as follows.  After a brief review of (2,2) toric GLSMs in section~\ref{s:basics},  we discuss their toric (0,2) deformations in section~\ref{s:toricdefs}, first describing classical bundle structure, and then turning to the quantum theory.  Having described the general structure, we illustrate it in a number of examples in section~\ref{s:examples}.  We end with an outlook on future directions.  The appendices deal with some technical aspects of the study.  In~\ref{app:nonlin} we present a non-renormalization theorem which shows that the so-called ``non-linear $E$-deformations'' do not affect the A/2 half-twisted theory.\footnote{This was developed in collaboration with M.R.~Plesser.}  Appendix~\ref{app:elgen} contains computations of elliptic genera which support our assertions regarding spontaneous SUSY breaking.  Appendix~\ref{app:bernstein} shows that for generic parameter values the number of solutions to the quantum sheaf cohomology relations for a NEF Fano variety $X$ is given by the Euler characteristic of $X$.  Finally, appendix~\ref{app:nonsusy} contains some details pertinent to the examples in section~\ref{s:examples}.

\acknowledgments  It is a pleasure to thank V.~Braun and E.~Sharpe  for useful discussions.  IVM would like to especially thank M.R.~Plesser for collaboration on the results presented in appendix~\ref{app:nonlin}.

RD is partially supported by NSF grant DMS 1304962.
ZL is supported by EPSRC grant EP/J010790/1.  IVM is grateful to the Mathematics Department at the University of Pennsylvania for hospitality while this work was undertaken; he is partially supported by NSF Grant DMS-1159404 and Texas A\&M.  
\section{Gauged linear sigma models and toric varieties} \label{s:basics}
The geometric object underlying our study is a smooth compact $d$-dimensional toric variety $X$.\footnote{A concise review of the relevant toric geometry and the relation to GLSMs may be found in~\cite{Kreuzer:2010ph}; a wealth of toric details may be found in~\cite{Cox:2011tv}.  While the assumption of smoothness simplifies the analysis, many of our results will generalize to simplicial toric varieties.}  The combinatorial data for $X$ is encoded in the toric fan $\Sigma_X \subset N_{\R} = N\otimes_{\Z} \R$, where $N$ is the lattice of one-parameter subgroups of the algebraic torus $T_N = (\C^\ast)^d \subset X$.\footnote{$X$ is smooth if and only if every $d$-dimensional cone in $\Sigma_X$ is simplicial and unimodular; it is compact if and only if $\Sigma_X$ generates $N_{\R}$.} The lattice of characters for the torus action is $M= N^\vee$, and we use $\la \cdot,\cdot \ra$ for the natural pairing.  We denote the integral generators of the one-dimensional cones in the fan by $\rho \in \Sigma_X$ and the full collection of $\rho$ by $\Sigma_X(1)$, with $n = |\Sigma_X(1)|$.  There are a number of ways to relate the combinatorial data to geometry; the one of most immediate use to us is based on the Cox homogeneous coordinate ring and the holomorphic quotient construction of $X$~\cite{MR1299003}.  Let us review the main aspects of this construction in three easy steps, at each step indicating the analogous step in the construction of an abelian gauged linear sigma model for $X$.
\begin{enumerate}
\item (coordinates and fields) To each $\rho \in \Sigma_X(1)$ we associate a coordinate $z_\rho$ on $\C^n$ and a generator of $S = \C[z_{\rho_1},\ldots,z_{\rho_n}]$. In the field theory this simply means we write down a Lagrangian for $n$ free (2,2) chiral superfields $\Phi_\rho$ with lowest components $z_\rho$.\footnote{Our superspace conventions will be those of~\cite{McOrist:2008ji,Kreuzer:2010ph}; we will need few details beyond the basic structure.}
\item (grading and gauge group)   
$\C^n$ admits a natural action of a ``big torus,''  $(\C^\ast)^{n}$, where for any $t\in (\C^\ast)^{n}$
$$t \cdot (z_{\rho_1},\ldots, z_{\rho_n}) = (t_{\rho_1} z_{\rho_1},\ldots, t_{\rho_n} z_{\rho_n}),$$
and this defines an abelian group $G_{\C}$ via the exact sequence\footnote{The reader may perhaps be familiar with the equivalent definition $G_{\C} = \Hom(\Pic(X), \C^\ast)$.}
\begin{align}
\xymatrix{ 1 \ar[r] & G_{\C} \ar[r] & (\C^\ast)^n \ar[r]^-{\rhot} & T_N \ar[r] & 1}~,\qquad
\rhot : t \mapsto (\prod_\rho t_{\rho}^{\rho^1},\ldots,\prod_\rho t_{\rho}^{\rho^d} )~.
\end{align}
For smooth $X$ $G_{\C} = (\C^\ast)^{k}$ with $k=n-d$, and clearly $G_{\C}$ inherits (from the big torus) an action on the $z_{\rho}$: $(\tau_1,\ldots,\tau_{k}) \cdot z_{\rho} = \prod_{a=1}^{k} \tau_a^{Q^a_\rho} z_\rho$.  The matrix of charges $Q^a_\rho$  is integral and grades the coordinate ring $S$. 

In the field theory we gauge the action of $G = \GU(1)^{k}$ on the $\Phi_\rho$ with charges $Q^a_\rho$.

\item (exceptional set and FI parameters)  The construction of $X$ as a holomorphic quotient requires one more ingredient --- the exceptional set $F$, a union of intersections of hyperplanes in $\C^n$:  for each minimal collection $\{\rho_i\}_{i\in I}$ that does not belong to a full-dimensional cone (sometimes known as a primitive collection; see, e.g.~\cite{Donagi:2011va,Donagi:2011uz}), the exceptional set $F$ includes $\cap_{i \in I} \{z_{\rho_i} = 0\}$.  When $\Sigma_X$ is simplicial, $X$ is presented as a geometric quotient
\begin{align}
X = \{ \C^n - F\} / G_{\C}~.
\end{align}
The gauge theory encodes the choice of exceptional set via the choice of Fayet-Iliopoulos parameters $\br = (r^1,\ldots, r^{k})$ for the $k$ abelian vector multiplets.  More precisely, there is a cone $\cC_{\text{cl}} \subset \R^{k}$ generated by the $n$ columns of the charge matrix $Q^a_\rho$ where the classical $D$-terms 
\begin{align}
D^a = \sum_{\rho} Q^a_\rho |z_\rho|^2 - r^a = 0~,\qquad a=1,\ldots,k,
\end{align}
have a non-empty solution set for $z_\rho$.  For compact $X$ $\cC_{\text{cl}}$ is a pointed---also known as strongly convex---cone.  $\cC_{\text{cl}}$ is further subdivided into full-dimensional cones (these are ``phases'' in the usual GLSM language), giving $\cC_{\text{cl}}$ the structure of a fan --- this is the secondary fan associated to $X$.  Each phase corresponds to a choice of exceptional set, and, in particular, there is a cone $\cC_X$ such that when $\br$ is in the interior of $\cC_X$ (denoted by $\inter(\cC_X)$ in what follows) then the symplectic quotient $\{(D^a)^{-1}(0)\} /\!\!/G$ is isomorphic as a K\"ahler manifold to $X$; the remaining phases correspond to other toric varieties birational to $X$.  
\end{enumerate}
We should mention one more key geometric structure.  For smooth $X$ the group of divisors modulo linear equivalence is generated by the toric divisors $D_\rho$ that correspond to images under the quotient of the hypersurfaces $\{z_\rho = 0\}$.  $H^2(X,\C) = \Pic(X)\otimes_{\Z} \C$ is generated by $\xi_\rho$, the classes dual to the $D_\rho$, which can be expanded in an integral basis $\{\eta_1,\ldots,\eta_{k}\}$ for $\Pic(X)$ as $\xi_\rho = \sum_a Q^a_\rho \eta_a$.  In particular, the canonical class is given by 
\begin{align}
K_X = -\sum_\rho \xi_\rho = -\sum_a \Delta^a\eta_a~,
\end{align}
where $\Delta^a = \sum_\rho Q^a_\rho$.  A smooth complete toric variety $X$ is NEF Fano if and only if $\Delta^a$ lies in the closure of $\cC_X$;\footnote{A variety $X$ is NEF Fano if and only if $X$ is complete, and the anti-canonical divisor  is NEF, i.e. has a non-negative intersection with every curve in $X$.}  $X$ is Fano if and only if $X$ is complete and $\Delta^a$ lies in the interior of $\cC_X$.  For later convenience we will fix the notation $W = H^2(X,\C)$.

\subsection{(2,2) toric GLSMs : quantum aspects}
Having set the notation and reviewed the basic correspondence between gauge theory and geometric data, we now turn to some aspects of the quantum theory.  The gauge theory has a dimensionful coupling constant $e$, meaning that the Lagrangian presentation in terms of chiral fields $\Phi_\rho$ coupled to vector multiplets $V_a$ is a good description at energy scales $\mu\gg e$.  Classically, when the parameters $r^a$ are deep in $\inter (\cC_X)$ we can reliably integrate out the gauge fields and obtain a description of the light degrees of freedom as a (2,2) NLSM with target space $X$.   

 Quantum mechanically the story is modified in a number of important ways.  First, the FI parameters naturally combine with the $\theta$ angles into complex parameters $q_a = e^{-2\pi r^a + i\theta^a}$ in a twisted superpotential.  In the NLSM description these correspond to the complexified K\"ahler parameters of the geometry.  Moreover, the FI parameters are not invariant under the RG flow.  Rather, under a change of scale $\mu_0\to \mu$ we find a one-loop running of the holomorphic couplings $q_a$ :
 \begin{align}
 \label{eq:1loop}
 q_a (\mu) = q_a(\mu_0) \left(\frac{\mu_0}{\mu}\right)^{\Delta^a}~,\qquad \Delta^a = \sum_{\rho} Q^a_\rho~.
 \end{align}  
This  has a direct analogue in the geometry: the running of the complexified K\"ahler parameters is determined by the anti-canonical class, which explains the reason for our insistence on $X$ being NEF Fano.  The positivity of $-K_X$ ensures that we can choose parameters such that in the UV limit $X$ will be driven to a smooth large radius limit, so that we can reliably match the UV physics of the GLSM and the NLSM for a large smooth manifold.  We will make a few comments on more general toric $X$ in section~\ref{s:discuss}.

Another important modification is an IR puzzle whose resolution is also key to most of the remarkable simplifications offered by the GLSM~\cite{Witten:1993yc,Morrison:1994fr}.  Since $\cC_{\text{cl}}$ is a pointed cone and $X$ is NEF Fano, the RG flow will eventually drive the FI parameters to a regime where classically supersymmetry is broken, even though the theory has a non-trivial Witten index---the Euler number of $X$.  The resolution of the puzzle is provided by the $\sigma$-vacua.  Recall that each (2,2) vector multiplet $V_a$ contains a complex scalar $\sigma_a$,\footnote{The common mathematics notation of $\sigma$ for cones in various fans is in delightful conflict with the standard physics notation of the $\sigma$ fields.  We hope the context makes it clear which one is meant.}
and the classical Lagrangian includes a term 
\begin{align}
\label{eq:sigmaz22}
\cL \supset 2\sum_{\rho} |z_\rho|^2 \left| \textstyle\sum_{a} Q^a_\rho \sigma_a \right|^2~,
\end{align}
indicating that a large $|\sigma/\mu|$ expectation value induces a mass for the chiral fields.  Assuming that the mass is large, we can integrate out the $\Phi_\rho$ at one loop and obtain an effective superpotential for the twisted chiral superfields $\Sigma_a = \sigma_a+\ldots$.  The critical points of this superpotential are solutions to
\begin{align}
\label{eq:QCR}
\prod_\rho \left(\mu^{-1} \textstyle\sum_{b=1}^{k} Q^b_\rho \sigma_b\right)^{Q^b_\rho}  = q_a(\mu) ~.
\end{align}
The approximation of large $\sigma$ vevs is self-consistent for NEF Fano $X$ when $q_a(\mu) \to \infty$, while when $q_a(\mu) \to 0$, and the NLSM is a good description of the low energy physics, the $\sigma$ vevs are small.  Hence, the ground states are reliably described by a set of massive ``Coulomb'' vacua labeled by solutions to~(\ref{eq:QCR}), while in the UV the NLSM provides a good description.\footnote{There are a number of reasons to think that the $\sigma$-vacua exhaust the supersymmetric ground states.  For instance, in examples it is easy to ascertain that the number of solutions to~(\ref{eq:QCR}) counted with multiplicity is indeed given by the Euler number of $X$.  The general proof for smooth NEF Fano $X$ is given in appendix~\ref{app:bernstein}.  A general discussion of Higgs and Coulomb vacua in various GLSM phases may be found in~\cite{Melnikov:2006kb}.} 

\subsection{A geometric interpretation}
The physics is related to a number of  geometric structures.  The connection is most easily made in the context of the A-model topological field theory associated to the GLSM, where one uses the non-anomalous vectorial $\GU(1)$ R-symmetry to twist the theory.  In this case we can forget about the RG running and the scale $\mu$ and concentrate on a set of local topological observables --- the operators $\sigma_a$, their relations and correlators.  In the geometric phase the $\sigma_a$ correspond to the generators $\eta_a$ of $H^{1,1}(X)$; the chiral ring relations~(\ref{eq:QCR}) are the quantum cohomology relations\footnote{For Fano $X$ they reduce to the Stanley-Reisner relations in the classical $q\to 0$ limit.}; and correlators $\la \sigma_{a_1} \cdots \sigma_{a_s} \ra$ are generating functions for genus zero Gromov-Witten invariants~\cite{MR1265307,Morrison:1994fr}.  The Coulomb branch description offers the most economical way to arrive at the quantum cohomology relations and to compute correlators without performing sums over gauge instantons~\cite{Melnikov:2006kb}.

\section{Toric (0,2) deformations} \label{s:toricdefs}
Having reviewed the relevant (2,2) structures, we will now turn to our real interest:  the (0,2) deformations.  As in the previous section, we will first discuss some geometrical aspects and their interpretation in terms of GLSM data.  We will then discuss some key quantum issues that will enable us to describe the (0,2) GLSM parameter space.

\subsection{Monadic and toric deformations}
The tangent sheaf of any projective toric variety can be described via the Euler sequence, an exact sequence of sheaves
\begin{align}
\label{eq:Euler}
\xymatrix{0 \ar[r] & W^\ast\otimes \cO_X \ar[r]^-{\bE} & \oplus_\rho \cO_X(D_\rho) \ar[r] & T_X \ar[r] & 0}~,
\end{align}
where the map $\bE$ is an $k\times n$ matrix $\bE^a_\rho = Q^a_\rho z_\rho$~.  This explicit presentation offers a way to obtain a family of deformed bundles $\cE \to X$ simply by deforming the maps $E^a_\rho$:
\begin{align}
\bE^a_\rho \to Q^a_\rho z_\rho + s^a_\rho(z)~,
\end{align}
where $s^a_\rho \in H^0(X,\cO_X(D_\rho))$. 

This presentation has a number of advantages to other approaches.  For instance, we might begin more abstractly by first computing $H^1(X,\End T_X)$ to obtain a parametrization of the first order deformations.  This is complicated even in relatively simple examples~(see~\cite{DonagiLuIVM2} for some explicit computations); moreover, there may be higher order obstructions to some of these first-order deformations.  By contrast, the ``monadic deformations'' obtained by deforming the Euler sequence are easy to count and are obviously unobstructed.\footnote{A monad bundle is one that can be expressed as a cohomology of a three-step complex.   These are the natural bundles obtained in GLSM constructions, where each step is a sum of toric line bundles. These are well-suited to studies via computational algebraic geometry.  The tangent bundle is a particularly simple example: the complex only involves two bundles.}

Not all unobstructed deformations of $T_X$ are monadic.  For instance, for $X= dP_3$ $h^1(\End T_X) = 15$ and $h^2(\End T_X) = 0$, so that we expect a $15$-parameter family of deformed bundles~\cite{DonagiLuIVM2}; however, only three of those parameters are captured by deforming the Euler sequence.

$T_X$ is a torus-equivariant bundle.  For an extreme example, for $X = \P^1$ and $T_X = \cO(2)$, and the total space is a toric variety, with a fan consisting of the two cones $\sigma_1$ and $\sigma_2$ in $N_{\R}$, as illustrated in figure~\ref{fig:TXP1}:\footnote{A torus-equivariant bundle is a toric variety if and only if it is a sum of line bundles~\cite{Cox:2011tv}, so typically the total space of $T_X$ will not be a toric variety.}
\begin{figure}[t]
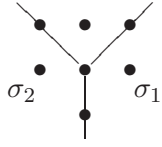

\[
\xy (0,0)*\xybox{<0.6mm,0.0mm>:
(0,0)*{\bullet}="0";
(10,10)*{\bullet}="1"; %(17,17)*{N_{\R}};
(-10,10)*{\bullet}="2"; 
(14,-5)*{\sigma_1};
(-14,-5)*{\sigma_2};
(10,0)*{\bullet};
(0,10)*{\bullet};
(-10,0)*{\bullet};
%(0,5)*{\bullet}="3"; (-1,8)*{3};
(0,-10)*{\bullet}="4";
{\ar@{-}"0";(15,15)};
{\ar@{-}"0";(-15,15)};
{\ar@{-}"0";(0,-15)};
%{\ar@{.}"1";"3"};{\ar@{.}"3";"4"};{\ar@{.}"4";"1"};
}
\endxy
\]
\caption{Tangent bundle of $\P^1$.}
\label{fig:TXP1}
\end{figure}
General monadic deformations do not respect the toric structure;  to preserve torus equivariance we need 
$\bE^a_\rho = E^a_\rho z_\rho$, i.e. the deformation simply replaces $Q^a_\rho$ with a more general complex-valued matrix $E^a_\rho$.
There is a class of Fano $X$, known as ``plain'' toric varieties, for which $H^0(X,\cO_X(D_\rho))$ is just one-dimensional, so that every monadic deformation is toric~\cite{Kreuzer:2010ph}.\footnote{The combinatorial condition for this is that the facets of the reflexive polytope $\Delta$ associated to $\Sigma_X$ contain no points in their relative interior.} For instance, $dP_3$ is plain, while $\P^1\times \P^1$ is not.  However, we can show that for any projective toric variety every torus-equivariant (or toric for short) deformation of $T_X$ is monadic~\cite{DonagiLuIVM2}.  To summarize, we have the following containments for deformations of the tangent sheaf:
\begin{align*}
\text{toric} \subseteq \text{monadic} \subseteq \text{all deformations}~.
\end{align*}
The first containment is an equality for plain Fano toric varieties.  We do not know the conditions under which the second containment becomes an equality.

The deformed Euler sequence defines a sheaf that can fail to be locally free on some subvariety in $X$.  A bundle is obtained if and only if $\bE$ has full rank at every point on $X$; while this is assured for sufficiently small deformations around $E^a_\rho = Q^a_\rho$, it will fail for a sufficiently large deformation.  A particularly drastic degeneration occurs when the matrix $E^a_\rho$ drops rank, so that $\bE$ fails to have full rank at every point on $X$.

\subsection{$E$ couplings in the GLSM}
The preceding geometric structure encodes the monadic (0,2)--preserving deformations of the (2,2) GLSM for $X$.  Splitting up the (2,2) chiral and twisted chiral multiplets into (0,2) ones, we have
\begin{align}
\Phi_{\rho} &\to Z_\rho~,\Gamma_\rho~&
\Sigma^{(2,2)}_a &\to \Sigma_a~, \Upsilon_a~,
\end{align}
where $Z_\rho = z_{\rho} + \ldots$ and $\Sigma_a = \sigma_a + \ldots$ are (0,2) bosonic chiral superfields, while $\Gamma_\rho$ and $\Upsilon_a$ are fermi (0,2) superfields with lowest component a left-moving Weyl fermion; $\Upsilon_a$ contains the curvature of the $a$-th gauge field.  While $\Upsilon_a$ is chiral, i.e. it is annihilated by the $\cDb$ superspace derivative, $\Gamma_\rho$ satisfies a more general constraint:
\begin{align}
\cDb \Gamma_\rho = i \sqrt{2} \sum_{a=1}^{k} \Sigma_a Q^a_\rho Z_\rho~.
\end{align}
The monadic (0,2) deformations simply replace this with a more general coupling consistent with the gauge symmetries and (classical) $\GUL$ symmetry which assigns charges $-1$ to $\Gamma_\rho$ and $\Sigma_a$ and leaves all other fields invariant:
\begin{align}
\label{eq:Ecoupling}
\cDb \Gamma_\rho = i\sqrt{2} \sum_{a=1}^{k} \Sigma_a \bE^a_\rho(Z_\rho)~.
\end{align}
A few comments are in order.
\begin{enumerate}
\item $\bE^a_\rho$ is a chiral field in order to be consistent with $\cDb^2 = 0$.
\item $\bE^a_\rho$ is part of the holomorphic data of this supersymmetric theory.  This is particularly clear for toric GLSMs: up to relabeling $\Gamma$ and its conjugate $\Gammab$, an isomorphic theory is obtained by using chiral $\Gamma^\rho$ and a (0,2) superpotential 
$$\cL_{\cW} \supset i\sqrt{2} \int d\theta \sum_{\rho,a} \Gamma^\rho \Sigma_a \bE^a_\rho(Z_\rho).$$ 
This point of view is used in appendix~\ref{app:nonlin} to show that smooth A/2 theories are independent of terms that are non-linear (in $Z_\rho$) in the $\bE^a_\rho$.
\item The interaction leads to a scalar potential contribution generalizing~(\ref{eq:sigmaz22}) to
\begin{align}
\label{eq:sigmaz02}
\cL \supset 2\sum_{\rho}\left| \textstyle\sum_{a} \bE^a_\rho(z) \sigma_a \right|^2~.
\end{align}
\item We do not include terms of higher order in $\Sigma_a$ in order to preserve the asymptotics of the scalar potential and selection rules that follow from the $\GUL$ symmetry.  In particular, this means that for generic $\bE^a_\rho$ integrating out the gauge degrees of freedom and massive $\Sigma_a$ multiplets leads to low energy degrees of freedom described by a (0,2) NLSM for the monad bundle $\cE \to X$ defined by the deformed Euler sequence~(\ref{eq:Euler}).
\item $\GUL$ is anomalous, and only a $\Z_{\gcd(\Delta^1,\ldots,\Delta^k)}$ subgroup remains as a symmetry. 
\end{enumerate}

\subsection{(0,2) SUSY breaking and $\sigma$-vacua}
The (0,2) GLSM defined by the $E$-deformation has a classical $\GUL\times\GUR$ symmetry and (0,2) supersymmetry.  While either $\GUL$ or $\GUR$ is anomalous, there is a vectorial R-symmetry $\GU(1)_{\text{V}}$.  This symmetry ensures that the classical (0,2) SUSY algebra is undeformed, so that it is sensible to ask whether (0,2) SUSY is spontaneously broken.\footnote{The most general deformation of the (0,2) SUSY current algebra was derived in~\cite{Dumitrescu:2011iu}. } 

As in the (2,2) theory, examination of the classical theory suggests both a SUSY puzzle and its resolution~\cite{McOrist:2007kp}.  While for $\br \in \cC_{\text{cl}}$ we find classical SUSY vacua, when $\br \not\in \cC_{\text{cl}}$ it appears that SUSY is broken.  Of course this cannot happen for generic parameter values due the non-vanishing Witten index and elliptic genus,\footnote{The latter can now be computed directly for  large classes of (2,2) and (0,2) GLSMs~\cite{Benini:2013nda,Gadde:2013dda,Benini:2013xpa} .} and the resolution is again provided by the ``Coulomb vacua'' where the $\sigma_a$ acquire expectation values and give large masses to the fields $Z_\rho,\Gamma^\rho$.  The latter can be integrated out, yielding an effective superpotential for the $\Sigma_a$ and $\Upsilon_a$ multiplets.  The general monadic deformations can be decomposed as follows:
\begin{align}
\bE^a_\rho = \sum_{\rho'} E^a_{\rho,\rho'} z_{\rho'} + F^a_\rho(z)~.
\end{align}
A non-zero linear term $E^a_{\rho,\rho'}$ is only allowed if $z_{\rho}$ and $z_{\rho'}$ have the same gauge charges, and $F^a_\rho$ is composed of any non-linear monomials in the $z$ that carry the same charge as $z_{\rho}$.  The toric deformations correspond to $E^{a}_{\rho\rho'} = E^a_{\rho} \delta_{\rho,\rho'}$ and $F^a_\rho = 0$.   In that case, it is easy to specialize the result of~\cite{McOrist:2007kp} to obtain the effective (0,2) superpotential $\int d\theta \Upsilon_a J_a(\Sigma)$, with
\begin{align}
J_a = -\frac{1}{8\pi i} \log \left[ q_a^{-1} \prod_\rho \left( \mu^{-1} \sum_{b=1}^{k} E^b_\rho \sigma_b \right)^{Q^a_\rho}\right]~.
\end{align}  The supersymmetric vacua are then zeroes of $J_a(\sigma) = 0$, or, equivalently, the solutions to the quantum sheaf cohomology relations (QSCR)
\begin{align}
\label{eq:QSCR}
\prod_\rho \left( \mu^{-1} \textstyle\sum_{b=1}^{k} E^b_\rho \sigma_b \right)^{Q^a_\rho} = q_a(\mu)~.
\end{align}
As a consequence, we see that when $\rank(E) < k$, on one hand, a $\Sigma$ multiplet will be free in the IR; on the other hand, for generic $q_a$ (0,2) SUSY will be broken, as~(\ref{eq:QSCR}) will be over-determined.  
More precisely, suppose that the rows of $E^a_\rho$ are linearly dependent---without loss of generality we may assume $E^1_\rho = \sum_{a=2}^k x_a E^a_\rho$.  Setting $\sigmat_a = \sigma_a + x_a \sigma_1$ we then have to solve the $k$ equations
\begin{align}
\prod_\rho \left( \mu^{-1} \textstyle\sum_{b=2}^{k} E^b_\rho \sigmat_b \right)^{Q^a_\rho} = q_a(\mu)~
\end{align}
for the $k-1$ variables $\sigmat_b$.  It is clear that these will not have solutions for generic $q_a$.  A solution exists for special choices of the $q_a$, but the IR dynamics, and in particular the topological heterotic ring, of such a theory will inevitably be singular due to the decoupled $\Sigma$ multiplet.  Since for every such choice of parameters both the classical geometry and the quantum theory are singular, we may as well not include these in our discussion of the parameter space.  In what follows we will therefore focus on full rank E-deformations.  As we will now see, this restricted parameter space still leads to a rich phenomenology.

\subsection{The spectral cover} \label{ss:speccover}
When $E$ has full rank, the QSCR may still degenerate in a number of ways.  Let $\cP=\{(E,q)\}$ denote the set of GLSM parameters.  It is useful to think of the solutions to the QSCR in terms of a spectral cover
\begin{align}
\cS \subset  \{(E,q,\sigma)\} \overset{\pi} \longrightarrow \cP~.
\end{align}
We consider the solution set $\cS$  as a subvariety in $\{(E,q,\sigma)\}$ and illustrate the projection in figure~\ref{fig:cover}.  Given a point $p \in \cP$, we distinguish the following possibilities for $\pi^{-1}(p)$.
\begin{figure}[t]
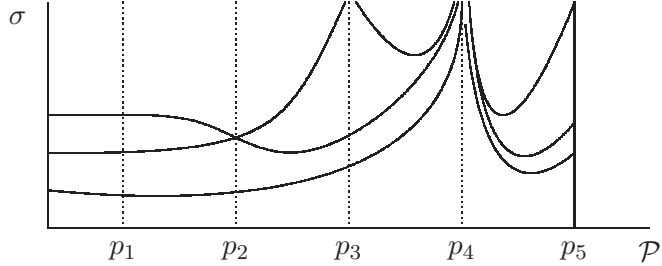

\[
\xy
%%%axes and labels, and the lines p1, p2,p3, p4,p5
(0,0)*{}="0";
{\ar@{-}"0";(80,0)};(80,-3)*{\cP};
{\ar@{-}"0";(0,30)};(-4,28)*{\sigma};
{\ar@{.}(10,0);(10,30)};(10,-3)*{p_1};
{\ar@{.}(25,0);(25,30)};(25,-3)*{p_2};
{\ar@{.}(40,0);(40,30)};(40,-3)*{p_3};
{\ar@{.}(55,0);(55,30)};(55,-3)*{p_4};
{\ar@{-}(70,0);(70,30)};(70,-3)*{p_5};
%%%%%
%%%%% curve construction:  note the ~pC in the commented out versions it helps to 
%%%%% make the curve go where you want by controlling the spline with visible labels
%%%%%
(0,5);(55,30)**\crv{(5,7),(55,0)}; 
%(0,10);(39.5,30)**[green]\crv~pC{(10,10)&(30,10)&(35,20)};
(0,10);(39.5,30)**\crv{(10,10)&(30,10)&(35,20)};  %%  
%(0,15);(54.5,30)**[blue]\crv~pC{(10,15)&(21.6,16)&(33,5)&(52,20)};
(0,15);(54.5,30)**\crv{(10,15)&(21.6,16)&(33,5)&(52,20)};
(41,30);(54,30)**\crv{(49,15),(50,16)};
(55.5,27);(70,10)**\crv{(58,0)};
(56.0,30);(70,14)**\crv{(58,0)};
(56.0,30);(70,30)**\crv{(58,0)};
\endxy
\]
\caption{Spectral cover with some representative points.}
\label{fig:cover}
\end{figure}
\begin{enumerate}
\item At a generic point, like $p_1$, the pre-image consists of a number of isolated points; this is the degree of the spectral cover.  In NEF Fano examples this number is given by $\chi(X)$.   As discussed in appendix~\ref{app:bernstein}, this follows by applying Bernstein's theorem for counting $\C^\ast$ solutions of Laurent polynomial systems to the QSCR. 
\item $p_2$ illustrates a ramification point, where two solutions merge --- this is a perfectly smooth point as far as the A/2 theory is concerned.\footnote{Such smooth behavior at ramification points where SUSY vacua merge is perhaps familiar from the $B$-twisted (2,2) Landau-Ginzburg theories~\cite{Vafa:1990mu}:  the classic example is the superpotential $W  =  X^{n+1} - t X$ for a chiral superfield $X$, where setting $t = 0$ merges the $n$ SUSY vacua and yields a smooth theory---a minimal model.}
\item The spectral cover fails to be proper at $p_3$, where a solution has wandered off to infinity --- more precisely, the $\pi$--pre-image of a compact set that includes $p_3$ is not compact.  In this case some A/2 correlators will diverge, leading to a singular theory.  The union of all such points yields the singular locus $\cA_{\tsing} \subset \cP$.
\item The spectral cover is not surjective at $p_4$, where all solutions have wandered off to infinity.  The result is spontaneous SUSY breaking for some locus $\cA_{\nSUSY} \subset \cP$.
\item The spectral cover may have a positive dimensional vertical component, as at $p_5$, where a $\sigma$-branch has opened up.  We will denote the locus of such points by $\cA_+ \subset \cP$.  Clearly $\cA_+$ will intersect the closure of $\cA_{\nSUSY}$. 
\end{enumerate}
The points $p_{3,4,5}$ inevitably lead to singular $A/2$ theories --- the correlators diverge as these points are approached.  This is easily seen from~\cite{McOrist:2007kp}, where every A/2 correlator $\la \sigma_{a_1} \cdots \sigma_{a_s}\ra$ is presented as a weighted sum over the solution set to the QSCR.\footnote{  The weight takes the form $\sigma_{a_1} \cdots \sigma_{a_s}$ times a measure factor derived from the superpotential; the form is such that even if some subset of correlators remain finite as a $\sigma_a$ tends to infinity, there will be correlators that diverge.}
We will say that  a point $p \in \cP$ is singular if and only if the A/2 theory is singular; otherwise, we will say $p$ is smooth.  The union of all of these singular points defines the discriminant of the A/2 theory.  For $q_a\in\C^\ast$ the discriminant is a variety that can be described explicitly by applying Bernstein's theorem.  More details are provided in appendix~\ref{app:bernstein} and the references. 

\subsection{A comment on SUSY breaking}
While the preceding statements regarding the A/2 theory and the QSCR are rigorous, we should make a small caveat regarding spontaneous SUSY breaking in the physical theory.
We are claiming to understand the SUSY vacua of a strongly coupled theory based on the effective action for the $\Sigma$ multiplets.  Experience with many examples, for instance the earlier work of~\cite{Tong:2008qd,Gadde:2013dda}, and many checks show that this is quite reasonable, but it is not rigorous, since we do not have good control of the $\Sigma$ kinetic term.  As a check of the proposal, we can compute the flavored elliptic genus of the model following~\cite{Benini:2013nda,Gadde:2013dda,Benini:2013xpa}.  As we show in appendix~\ref{app:elgen}, at least in a simple example we find a consistent picture:  the elliptic genus is non-zero in the SUSY case and vanishes in the non-SUSY case.  

\subsection{Redefinitions, bundles, and quotients: classical GLSM and geometry}
We have learned that toric (0,2) SUSY deformations are characterized by a rank $k$ matrix $E^a_\rho$.  Naively, this would suggest that the deformation space of $T_X$ is described by $nk$ complex parameters.  Geometrically, this is clearly an over-count.  Some of these parameters are simply a choice of basis on $\cO_X \otimes W$, while others can be undone by toric automorphisms of $X$.  A similar picture emerges in matching the GLSM to the (0,2) NLSM.  Many of the $E^a_\rho$ parameters can be absorbed into holomorphic field redefinitions that only modify the presumably irrelevant kinetic terms.  While these may change the details of the metric on $\cE \to X$, we do not expect them to change the universality class, i.e. we expect theories with $E^a_\rho$ related by holomorphic field redefinitions to have the same IR physics. 

The field redefinitions, since we restrict attention to toric deformations, should also be torus--equivariant.  This means we consider redefinitions
\begin{align}
Z_\rho &\to u_\rho Z_\rho~,&
\Gamma^\rho & \to v_\rho \Gamma^\rho~,&
\Sigma_a & \to \sum_{b}  \Sigma_bR_{a}^b~,
\end{align}
where $u_\rho,v_\rho \in \C^\ast$, and $R_a^b \in\GL(k,\C)$. The induced action on $E^a_\rho$, thought of as an $k\times n$ matrix, is fixed from the form of the $E$-couplings in the Lagrangian~(\ref{eq:Ecoupling}).  The result is
\begin{align}
\label{eq:Eaction}
E \mapsto R E T,\qquad 
T = \begin{pmatrix} 
t_{\rho_1} & 0 			& \cdots &0 \\
0		&t_{\rho_2}	& \cdots & 0 \\
\vdots	& \vdots		&\ddots & \vdots \\
0		& 0			& \cdots & t_{\rho_n}
\end{pmatrix}~,
\end{align}
where $t_\rho = u_\rho v_\rho^{-1}$.  We see that not all of the redefinitions act effectively on the $E$-couplings.  For instance, the $u_\rho$ and $v_\rho$ only appear via the combinations $u_\rho v_\rho^{-1}$, and, in addition, setting $R = t^{-1} \iden_{k\times k}$ and $T= t \iden_{n\times n}$ leaves the $E$ couplings invariant.  This is a consequence of the classical $\GUL$ symmetry of the GLSM Lagrangian, which assigns charge $-1$ to the $\Gamma$ and $\Sigma$ superfields and leaves other fields invariant.  %%We will call the complexification of this classical symmetry action $\CL$. 

%\subsubsection*{Deformations around generic $E$}
We now have a simple presentation of the $E$ deformation space: we should simply consider the $k\times n$ matrices $E$ modulo the equivalence relations $E \sim RET$.    It is easy to see that expanding the action of the redefinitions around a generic $E$ matrix leads to a a $k(n-k)-n+1$--dimensional space of infinitesimal deformations.  First, we write
\begin{align}
E = \begin{pmatrix} L & | & L S \end{pmatrix}~,
\end{align}
where $L \in \GL(k,\C)$ and $S$ is a $k\times d$ matrix.  Writing $T = A^{-1} \oplus B$, where $A$ is a diagonal $k\times k$ matrix, and $B$ is a diagonal $d \times d$ matrix, we see that we can use the $R$ redefinitions to bring $E$ to a canonical form.  Setting $R = A L^{-1}$, we find
\begin{align}
E \mapsto \begin{pmatrix} \iden_{k\times k} & | & A S B \end{pmatrix}~.
\end{align}
Finally, we can use the remaining $n-1$ effective parameters in $A$ and $B$ to bring $S$ into a canonical form
\begin{align}
S = \begin{pmatrix}  1 & 1 & \cdots & 1 \\ 1 & s_{22} & \cdots & s_{2d} \\ \vdots & \vdots & \vdots \\ 1 & s_{k2} & \cdots & s_{kd} \end{pmatrix}~,
\end{align}
leaving $ (k-1)(d-1) = k(n-k) -n +1$ complex parameters.  

\subsubsection*{A quotient of a Grassmannian}
 Since we work with a full rank $E$, the first quotient via the identification $E \sim R E$ simply gives us the complex Grassmannian $\Gr(k,n)$.  This compact K\"ahler manifold of dimension $k(n-k)=kd$ admits a transitive $\GL(n,\C)$ action.  The $T$ equivalence is a quotient by  a $T_{n-1} = (\C^\ast)^{n-1}$ subgroup of $\GL(n,\C)$, so the remaining quotient is $\Gr(k,n) / T_{n-1}$.  As we will illustrate and discuss below, the $T_{n-1}$ action has non-trivial stabilizers in $\Gr(k,n)$, and the quotient topology fails to separate orbits.  This quotient therefore defines a moduli stack for the toric deformations of $T_X$.

\subsection{Redefinitions in the quantum theory: the action on the K\"ahler parameters}
While the $q_a$ are classically invariant under the redefinitions, this is not true in the quantum theory because the left- and right-moving fermions transform differently under the rescalings of the fields.  While the transformations of the gauge-neutral $\Sigma$ multiplets do not affect the $q_a$, the change in the fermi measure induced by the $T_n$ action corresponds to rescaling 
\begin{align}
q_a \mapsto q_a \times \prod_{\rho} t_\rho^{Q^a_\rho} ~.
\end{align}
Note that the rescalings act equivariantly on~(\ref{eq:QSCR}); moreover, the complexified $\GUL$ action is equivalent to a change in the renormalization scale $\mu$ in~(\ref{eq:1loop}). 

For $q_a \in \C^\ast$ the field redefinition quotient of the GLSM parameter space yields what we will call the ``quantum moduli stack.''  The $\GL(k,\C)$ quotient yields the reduced A/2 model parameter space
\begin{align}
\label{eq:Pi}
\Pi_{k,n} = \cP/\GL(k,\C) =  \Gr(k,n) \times T_k~,
\end{align}
and the quantum moduli stack is then $\Pi_{k,n} /T_n$.   As in the previous section, this quotient has non-separated points associated to stabilizers for the $T_n$ action.  The moduli stack for the classical theory can be obtained by taking the classical $q_a \to 0$ limit.  This classical moduli stack will have points with larger stabilizers, so that in some sense the quotient is better behaved in the quantum case.  

More generally, we could think of the $q_a$ as local coordinates on the toric variety whose toric fan is the secondary fan of $X$ and study the full quotient.  Such a description would then automatically include all of the limiting points that describe the different phases of the GLSM.  We will not include these limiting singular loci in this work, because without them the resulting moduli spaces are much simpler, and we do not lose very much by leaving out these infinite distance points.  Since we can obtain the QSCR and A/2 correlators as rational functions of the GLSM parameters, it is a simple matter to evaluate their limits.

The stabilizer of the $T_n$ action on $(q_1,\ldots, q_k) $ has a simple combinatorial description. Since $X$ is smooth and projective, there is a short exact sequence
\begin{align}
\xymatrix{0 \ar[r] & M \ar[r]^-{\alpha} & \Z^{\Sigma_X(1)} \ar[r]^-{\beta} & \Pic(X) \ar[r] & 0~,}
\end{align}
where the first map is represented by a $n\times d$ matrix $A$ that encodes the coordinates of the primitive vectors $u_\rho \in N_{\R}$ that generate the one-dimensional cones $\rho \in \Sigma_X(1)$, and the second map is represented by our familiar charge matrix $Q$.  From this it follows that a point $(q_1,\ldots,q_k) \in (\C^\ast)^k$ is stabilized by the $T_d \subset T_n$ subgroup
\begin{align}
t_\rho = \prod_{\alpha=1}^d \tau_\alpha^{A^\rho_\alpha}~,
\end{align}
with $(\tau_1,\ldots,\tau_d) \in T_d$.  We will use this fact below.

In a sense we are now finished with a description of the toric moduli of the A/2 GLSM for the NEF Fano toric variety $X$: it is given by the quantum moduli stack.  This is not a very satisfactory description.  While it does include every smooth A/2 theory, the non-separated points complicate the geometric interpretation of the quotient.  A better description is obtained as follows.  

In order to obtain a separated variety from the quotient, and therefore a moduli space of A/2 theories, we must remove some of the $T_n$ orbits.  Since $T_n$ is a reductive group, we can use GIT to produce a separated quotient by choosing a stability condition and removing an exceptional set of ``unstable'' points~\cite{Mumford:GIT}.  This is of course very familiar from the toric constructions we have been discussing.  Our goal is to find a quotient compact variety $\cMb(X)$, which contains the set of smooth A/2 theories as an open subset $\cM(X) \subset \cMb(X)$, and GIT will produce such a variety for any stability condition where the exceptional set $F \subset \cP$ is contained in the discriminant locus $\cA_{\tsing} \subset \cP$.  We have not proven that this is the case for every NEF Fano $X$, but it is borne out in examples, as we will see in the next section.  Indeed, it seems that typically for any $X$ there are many acceptable stability conditions.  We leave the analysis of such acceptable conditions for general $X$ to future work and now turn to a study of some illuminating examples.

\section{Examples} \label{s:examples}
\subsection{The simplest example: $X=\P^1$}
The simplest example is provided by $X=\P^1$.  In this case $Q=(1~1)$ and $\bE = (E_1 z_1 ~ E_2 z_2)$.  Keeping the $E$ matrix full rank requires $E_1$ and $E_2$ not to vanish simultaneously; however, since~(\ref{eq:QSCR}) reduces to  
\begin{align}
E_1 E_2 \sigma^2 = q~,
\end{align}
we see that for $q\neq 0$ (0,2) SUSY requires $E_1,E_2 \in \C^\ast$.  In appendix~\ref{app:elgen} we show that this is consistent with the elliptic genus.  We can then use the rescalings to set $E_1 = E_2 =1$---the canonical (2,2) form.   Moreover, the complexified $\GUL$ rescaling can be used to set $q$ to any non-zero value, reflecting the running nature of this coupling, and the set of parameters is left invariant by a $\Z_2 \subset \GUL$, the unbroken global symmetry of the $\P^1$ model.  Geometrically all of this is not surprising: the tangent bundle of $\P^1$ is indeed rigid.

\subsection{$X=\P^1\times \P^1$}
Let $[z_1:z_2]$ and $[z_3:z_4]$ be the projective coordinates on $X$.  Then the charge matrix $Q$ and the $E$ matrix defining torus-equivariant deformations are 
\begin{align}
Q &= \begin{pmatrix} 1 & 1& 0 & 0 \\ 0 & 0 & 1 &1 \end{pmatrix}~, &
\bE & = 
\begin{pmatrix} 
E^1_{1} z_1 & E^1_{2} z_2 & E^1_{3} z_3 & E^1_{4} z_4 \\
E^2_{1} z_1 & E^2_{2} z_2 & E^2_{3} z_3 & E^2_{4} z_4
\end{pmatrix}~.
\end{align}
The QSCR are then
\begin{align}
\label{eq:QSCRP1P1}
(\sigma \cdot E_1) (\sigma \cdot E_2) & = q_1~,&
(\sigma \cdot E_3) (\sigma \cdot E_4) & = q_2~,
\end{align}
where we use a condensed notation $\sigma \cdot E_\rho = \sum_a \sigma_a E^a_{\rho}$.  As long as $q_a \in \C^\ast$, we immediately see that SUSY will be broken when a column of the $E$ matrix is identically zero.  On the other hand, for a generic $E$ matrix, the solution set to the QSCR will consist of $4$ points in $\C^2$.  For instance, when $\cE = T_X$, the QSCR are $\sigma_a^2 = q_a$, with roots $\{(\pm \sqrt{q_1},\pm \sqrt{q_2}),(\pm \sqrt{q_1},\mp\sqrt{q_2})\}$.  

We now characterize the points in $\cP$ according to the general discussion of the spectral cover for the QSCR in section~\ref{ss:speccover}.  Our first goal is to characterize the discriminant locus for $q_a\in \C^\ast$.  As we mentioned above, there is an explicit combinatorial presentation based on Bernstein's theorem---see appendix~\ref{app:bernstein} for more details, but in all of our examples we will be able to describe the discriminant without introducing this additional machinery.  Here we proceed as follows.  First, we compactify the $\C^2$ where the $\sigma_a$ take value to $\P^2$ by introducing homogeneous coordinates $[s_0:s_1:s_2]$ and setting $\sigma_{1,2} = s_{1,2} / s_0$.  The QSCR then lead to a subvariety $\Vb \subset \P^2$ defined by the vanishing locus of
\begin{align}
(s \cdot E_1) (s \cdot E_2) & = q_1s_0^2~,&
(s \cdot E_3) (s \cdot E_4) & = q_2s_0^2~.
\end{align}
When $\Vb$ is zero-dimensional, Bezout's theorem indicates that it consists of $4$ points counted with multiplicity.  There is a $1:1$ correspondence between finite solutions of the QSCR in~(\ref{eq:QSCRP1P1}) and points in $\Vb$ that do not lie in the compactification divisor $C = \{s_0 = 0\} \subset \P^2$.  In other words, the discriminant locus consists of the $(E,q)\in \cP$ for which $\Vb \cap C \neq \emptyset$, which holds if and only if
\begin{align}
D = \zeta_{13}\zeta_{14} \zeta_{23} \zeta_{24} = 0~,
\end{align}
where $\zeta_{ij}$ is the determinant of the minor of $E$ obtained by taking the $i$-th and $j$-th column.

This result is in accord with intuition from considering a gauge instanton expansion for the A/2 GLSM.  As in the (2,2) case, whenever $X$ is Fano, only a finite number of gauge instantons can contribute to any fixed correlator, leading to an expression polynomial in the $q_a$ with $E$-dependent coefficients.   Therefore, the only way for a correlator to diverge is if one of these coefficients diverges.   We therefore expect that the discriminant will have a $q_a$---independent component that can be computed in the classical $q_a\to 0$ limit, and, barring some jumping phenomenon for non-zero $q$, we expect the discriminant to be $q_a$--independent.  Indeed, $D = 0$ is precisely the locus where the rank of $E$ decreases for some point on $X$, and  $\cE$ fails to be locally free.

We see that for any choice of parameters where the A/2 theory is smooth, we obtain unbroken (0,2) SUSY in the physical theory.  On the other hand, at special singular points where $\Vb \subset C$, SUSY will be broken.  As the following two examples indicate, this can take place at either special or generic values of the $q_a$.
\begin{enumerate}
\item Set 
\begin{align*}
 E = \begin{pmatrix} 1 & 0 & 1 & 0 \\ 0 & 1 & 0 & 1 \end{pmatrix}~.
\end{align*}
This leads to QSCR $\sigma_1\sigma_2 = q_1$ and $\sigma_1 \sigma_2 = q_2$.  If $q_1 = q_2$, there is a continuum of SUSY vacua, and $\Vb$ degenerates to a conic in $\P^2$; if $q_1 \neq q_2$, then there are no SUSY vacua, and $\Vb$ consists of two points $[0:1:0]$ and $[0:0:1]$ in $C$.
\item Set
\begin{align*}
 E = \begin{pmatrix} 1 & 1 & 1 & 0 \\ 0 & 1 & 0 & 1 \end{pmatrix}~.
\end{align*}
Here the QSCR are $\sigma_1^2 = q_1 -q_2$ and $\sigma_1\sigma_2 = q_2$.  When $q_1\neq q_2$ there are SUSY vacua, but when $q_1 = q_2$ there are none:  $\Vb$ consists of $[0:1:0] \subset C$.
\end{enumerate}

\subsubsection*{The $\GL(2,\C)$ quotient and SUSY vacua}
We now make the first step in describing the moduli space for this A/2 theory.  As explained above, we restrict attention to full-rank $E$ and $q_a \in \C^\ast$.  In that case, the $\GL(2,\C)$ quotient by the $\Sigma$ field redefinitions, labeled by $R$ in (\ref{eq:Eaction}) reduces the $E$ parameters to points in $\Gr(2,4)$, a compact variety that we will represent by the vanishing of
\begin{align}
\label{eq:Pluecker}
P = \zeta_{12} \zeta_{34} + \zeta_{14} \zeta_{23} -\zeta_{13} \zeta_{24} 
\end{align}
in $\P^5$ with projective coordinates $[\zeta_{12},\zeta_{13},\ldots,\zeta_{34}]$.  Together with the $q_a$ this constitutes the reduced parameter space $\Pi_{2,4}$ of~(\ref{eq:Pi}).

The theory is singular on the reducible subvariety $\cA_{\text{sing}} = \{D = 0\} \subset \Pi_{2,4}$. In the Pl\"ucker coordinates its components are obtained by intersecting $P=0$ with one of the coordinate hyperplanes $\zeta_{12} = 0$, $\zeta_{13} = 0$, $\zeta_{14} = 0$, $\zeta_{24} = 0$.   

According to the general discussion in section~\ref{ss:speccover}, $\cA_{\text{sing}}$ will in general contain a subset $\cA_{\nSUSY}$, where there are no SUSY vacua.  $\cA_{\nSUSY}$ is a bit awkward to describe, since it is not in general closed.  However, we can describe its union with $\cA_{+}$, the locus where a continuum of $\sigma$ vacua opens up.  We leave the details to appendix~\ref{app:nonSUSYP1P1} and quote the answer here.  We have $\cA_+ = \cA'_+ \cap \{P=0\}$, where
\begin{align}
\cA'_+ = \{\zeta_{23} = 0,~\zeta_{14} = 0,~\Delta_+ = 0\} \cup \{\zeta_{13}=0,~\zeta_{24}=0,~\Delta_-=0\}~,
\end{align}
and $\Delta_\pm = q_2 \zeta_{12} \pm q_1 \zeta_{34}$.  Next, we have
\begin{align}
\cA_+ \cup \cA_{\nSUSY} = \{\{P=0\} \cap \cB \} \cup  \left\{ \{\P^2_1 \cup \P^2_2 \cup \P^2_3 \cup \P^2_4 \}\times (\C^\ast)^2\right\} ~,
\end{align}
where
\begin{align}
\label{eq:P2idef}
\P^2_i = \{\zeta_{ij} = \zeta_{ik} =\zeta_{il} = 0\}~, 
\end{align}
with $j,k,l$ distinct and not equal to $i$,
and 
\begin{align}
\cB = \{ \zeta_{14}\zeta_{23} = 0,~ \Delta_+ = 0\} \cup \{\zeta_{13} \zeta_{24} = 0~, \Delta_- = 0\}~.
\end{align}

\subsubsection*{The $T_4$ quotient and GIT}
Finally, we examine the action of the remaining field redefinitions on $\Pi_{2,4}$.  After a change of basis, the $T_4$ action on the $\zeta_{ij}$ and $q_a$ is described by the charge matrix
\begin{align}
\xymatrix@R=0mm@C=4mm{
~	&	\zeta_{12}	&\zeta_{13}	&\zeta_{14}		&\zeta_{23}	&\zeta_{24}	&\zeta_{34}	&q_1		&q_2\\
t'_1	&1			&1			&1			&0			&0			&0			&1		&0	\\
t'_2	&0			&1			&0			&1			&0			&1			&0		&1	\\
t'_3	&0			&1			&1			&-1			&-1			&0			&0		&0 	\\
t'_4	&0			&1			&-1			&1			&-1			&0			&0		&0	\\
}
\end{align}
In this basis it is clear that the $T_4$ action factors through a $T_2$ action on the $(\zeta_{13},\zeta_{14},\zeta_{23},\zeta_{24})$, and we can use $t'_{1,2}$ to set $q_{1,2} = 1$.  Moreover, there is a $\Z_2 \subset T_2$, generated by $(t'_3,t'_4) = (-1,-1)$ that does not act---this corresponds to the unbroken $\Z_2 \subset \GUL$ symmetry of this A/2 model.
A further change of basis yields the effective $T_2/\Z_2$ action
\begin{align}
\xymatrix@R=0mm@C=4mm{
~	&\zeta_{12}	&\zeta_{13}	&\zeta_{14}		&\zeta_{23}	&\zeta_{24}	&\zeta_{34}	&q_1		&q_2 & \\
%t_0	&1			&1			&1			&1			&1			&1			&0		&0    &\\		
\tau_1	&0			&1			&0			&0			&-1			&0			&0		&0 	&\\
\tau_2	&0			&0			&1			&-1			&0			&0			&0		&0	&.\\
}
\end{align}
Since we can eliminate $t_1$ and $t_2$ by fixing $q_1$ and $q_2$, let us focus on the $T_2/\Z_2$ action $(\tau_1,\tau_2)$.
A generic orbit is labeled by two invariants:  $u = \zeta_{13}\zeta_{24}$ and $v = \zeta_{14} \zeta_{23}$.  If we decide to keep all orbits of the $T_2/\Z_2$ action, then the result will be a non-separated space.  A hallmark of such pathology is an ambiguity in limits.  For instance, $u=0$ and $v\neq 0$ corresponds to the orbits with representatives
\begin{align}
(\zeta_{13},\zeta_{14},\zeta_{23},\zeta_{24}) \in \{ (1,v,1,0),~(0,v,1,1),~(0,v,1,0)\}~,
\end{align}
and the third orbit is in the closure of the first two.

We will say that an orbit is smooth(singular) if and only if it corresponds to a smooth(singular) A/2 theory.  The non-separated orbits are all singular, and since we are interested in giving a compactification of the moduli space of smooth A/2 theories, we can impose stability conditions to eliminate some of the singular orbits and obtain a separated compact quotient. 

Letting $\tau_0$ denote the $\C^\ast$ action on the projective coordinates of $\P^5$, we wish to impose stability conditions for the action 
\begin{align}
\xymatrix@R=0mm@C=4mm{
~	&\zeta_{12}	&\zeta_{13}	&\zeta_{14}		&\zeta_{23}	&\zeta_{24}	&\zeta_{34}	& \\
\tau_0	&1			&1			&1				&1			&1			&1			    &\\		
\tau_1	&0			&1			&0				&0			&-1			&0			 	&\\
\tau_2	&0			&0			&1				&-1			&0			&0				&.\\
}
\end{align}
This is just a three-dimensional toric quotient, and the possible stability conditions are in one-to-one correspondence with relative interiors of the cones in the secondary fan shown in figure~\ref{fig:stabforP1P1}.
\begin{figure}[t]
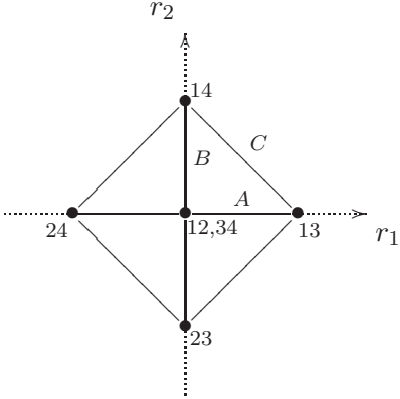

\[
\xy 
(0,0)*\xybox{<3.0mm,0.0mm>:
(0,0)*{\bullet}="1234"; (1.2,-0.7)*{{\scriptstyle{12,34}}};
(5,0)*{\bullet}="13"; (5.5,-0.7)*{\scriptstyle{13}};
(-5,0)*{\bullet}="24"; (-5.7,-0.7)*{\scriptstyle{24}};
(0,5)*{\bullet}="14"; (0.7,5.5)*{\scriptstyle{14}};
(0,-5)*{\bullet}="23";(0.7,-5.5)*{\scriptstyle{23}};
{\ar@{.>}(-8,0);(8,0)};
(9,-1)*{r_1};
{\ar@{.>}(0,-8);(0,8)};
(-1,9)*{r_2};
%{\ar@{-}"0";(0,10)};
%{\ar@{-}"0";(-10,-10)};
{\ar@{-}"1234";"13"^{\scriptstyle{A}}};{\ar@{-}"1234";"14"_{\scriptstyle{B}}};{\ar@{-}"1234";"24"};{\ar@{-}"1234";"23"};
{\ar@{-}"13";"14"_{\scriptstyle{C}}};{\ar@{-}"14";"24"};{\ar@{-}"24";"23"};{\ar@{-}"23";"13"};
};
\endxy
\]
\caption{Stability conditions for the $T_4$ quotient of the $\P^1\times\P^1$ model: the secondary fan. The moment maps  for the $(\tau_0,\tau_1,\tau_2)$ are $r_0,r_1,r_2$, and shown is the intersection of the fan with the $r_0 = 1$ plane.  We also label the two-dimensional cones $A$, $B$, $C$ as shown.}
\label{fig:stabforP1P1}
\end{figure}
There are many possibilities corresponding to the large collection of cones in the secondary fan, but clearly many give qualitatively similar features.   We focus on a representative set to discuss the possibilities for the exceptional set $F$.
\begin{enumerate}[(i)]
\item $\relint(ABC)$.  The moment maps satisfy $r_1 >0$, $r_2>0$, and $r_0 > r_1+r_2$. 
 $$ F= \{\zeta_{13} =0\} \cup \{\zeta_{14} = 0\} \cup\{\zeta_{12}=0,\zeta_{23}=0,\zeta_{24}=0,\zeta_{34}=0\}.$$
\item $\relint(A)$.  Now $r_2=0$, $r_1>0$, and $r_0 >r_1$, and
$$ F = \{\zeta_{13}=0\} \cup \{\zeta_{14}=0,\zeta_{23} \neq 0\} \cup \{\zeta_{14}\neq0,\zeta_{23} = 0\}\cup\{\zeta_{12}=0,\zeta_{23}=0,\zeta_{24}=0,\zeta_{34}=0\}.$$
\item $\relint(12)$.  This is a more interesting possibility, given by $r_{1,2} = 0$ and $r_0 > 0$. 
\begin{align*} F &= \{\zeta_{14}=0,\zeta_{23} \neq 0\} \cup \{\zeta_{14}\neq0,\zeta_{23} = 0\} 
\cup \{\zeta_{13}=0,\zeta_{24} \neq 0\} \cup \{\zeta_{13}\neq0,\zeta_{24} = 0\} \nonumber\\
&\qquad\cup\{\zeta_{12}=0,\zeta_{23}=0,\zeta_{24}=0,\zeta_{34}=0\}~.
\end{align*}
\item $\relint(C)$; $r_1>0$, $r_2 >0$, and $r_0 = r_1+r_2$, and
$$ F \supset \{\zeta_{13}=0\} \cup \{\zeta_{14} = 0\}\cup \{\zeta_{12}\neq 0\} \cup \{\zeta_{23} \neq 0\} \cup \{\zeta_{24} \neq 0\} \cup \{\zeta_{34} \neq 0\}~.$$
This is not a very interesting possibility, since $F$ contains all smooth orbits.  The same holds if we take $\relint(B)$, $\relint(13)$, or $\relint(14)$.
\end{enumerate}
In the first three cases the result is a compactification $\cMb(X)$ of the moduli space of smooth $A/2$ theories for $X = \P^1\times\P^1$.  Each compactification keeps some singular orbits and, in fact, some non-SUSY orbits!  The last assertion is easy to see since no exceptional set contains the union of the $\P^2_i$. In the first case the orbits that intersect $\cA_+$ are removed, while this is not the case for the remaining choices.

Let us analyze the first choice, $\cMb_{(i)}$, in detail.  Having removed the exceptional set, we may set $\zeta_{13} = \zeta_{14} =1$ to find
\begin{align}
\cMb_{(i)}(\P^1\times\P^1) = \{\zeta_{12}\zeta_{34} +u-v = 0\} \subset \P^3_{1122} [\zeta_{12},\zeta_{34}, u, v]~.
\end{align}
The discriminant is given by $\overline{\cA_\tsing} = \{uv = 0\} \subset \cMb_{(i)}$,  and the tangent bundle corresponds to the point $[0,0,1,1]$.\footnote{ The $\Z_2$ orbifold singularity reflects the fact that at the (2,2) locus this theory acquires an additional $\Z_2$ symmetry.  That is a special feature of product theories.}  Finally, the non-SUSY locus is just a collection of points:
\begin{align}
\overline{\cA_{\nSUSY}} = \{ [0,1,0,0],~ [q_1 , q_2, 0, \pm 1],~[q_1, -q_2, \pm 1, 0]\}~.
\end{align}

\subsection{$X = dP_1 = \F_1$}
For our next example, we consider the first del Pezzo, or, equivalently, the first Hirzebruch surface. Its fan is given in figure~\ref{fig:dP1}~, and the charge matrix is
\begin{align}
\label{eq:QdP1}
Q = \begin{pmatrix} 1 & 1 & 0 & 1 \\ 0 & 0 & 1 & 1 \end{pmatrix}~.
\end{align}
\begin{figure}[t]
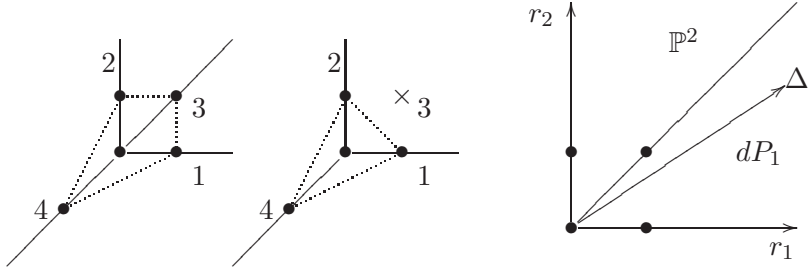

\[
\xy 
(-10,0)*\xybox{<1.5mm,0.0mm>:
(0,0)*{\bullet}="0";
(5,0)*{\bullet}="1"; (7,-2)*{1};
(5,5)*{\bullet}="3"; (7,4)*{3};
(0,5)*{\bullet}="2"; (-1,8)*{2};
(-5,-5)*{\bullet}="4";(-7,-5)*{4};
{\ar@{-}"0";(10,0)};
{\ar@{-}"0";(10,10)};
{\ar@{-}"0";(0,10)};
{\ar@{-}"0";(-10,-10)};
{\ar@{.}"1";"3"};{\ar@{.}"3";"2"};{\ar@{.}"2";"4"};{\ar@{.}"4";"1"};
};
(20,0)*\xybox{<1.5mm,0.0mm>:
(0,0)*{\bullet}="0";
(5,0)*{\bullet}="1"; (7,-2)*{1};
(5,5)*{\times}="3"; (7,4)*{3};
(0,5)*{\bullet}="2"; (-1,8)*{2};
(-5,-5)*{\bullet}="4";(-7,-5)*{4};
{\ar@{-}"0";(10,0)};
%{\ar@{-}"0";(10,10)};
{\ar@{-}"0";(0,10)};
{\ar@{-}"0";(-10,-10)};
{\ar@{.}"1";"2"};{\ar@{.}"2";"4"};{\ar@{.}"4";"1"};
};
(50,-10)*\xybox{<1.0mm,0.0mm>:
(0,0)*{\bullet}="0";
(10,0)*{\bullet}="1";
(10,10)*{\bullet};
(30,20)*{\Delta}="3";
(25,10)*{dP_1};
(15,25)*{\P^2};
%(20,20)*{\bullet}="3";
%(28,20)*{r_1=r_2};
(0,10)*{\bullet}="2";
{\ar@{->}"0";"3"};
{\ar@{->}"0";(30,0)};
{\ar@{-}"0";(30,30)};
{\ar@{->}"0";(0,30)};
(28,-3)*{r_1};
(-4,28)*{r_2};
};
\endxy
\]
\caption{$dP_1$ fan, its ``alternate phase'' ($\P^2$), and the secondary fan with $\Delta=-K_X$.}
\label{fig:dP1}
\end{figure}
The moment map equations are then
\begin{align}
&(i) ~~|z_1|^2+|z_2|^2+|z_4|^2 = r_1~, \quad
(ii) ~~|z_3|^2 + |z_4|^2 = r_2~, \nonumber\\
&\implies (iii)~~ |z_1|^2 +|z_2|^2 -|z_3|^2 = r_1-r_2~.
\end{align}
The exceptional collections can be seen from the moment map equations.  We find
\begin{align}
(i) \implies F \supset \{z_1=z_2=z_4=0\} \qquad
(ii) \implies F \supset \{z_3=z_4=0\}~,
\end{align}
and the remaining component depends on the sign of $r_1-r_2$.  When $r_1 >r_2$ the additional
component of $F$ is $\{z_{1}=z_{2} = 0\}$, and
\begin{align}
F  =  \{z_1=z_2=0\} \cup \{z_3=z_4=0\} \implies X = dP_1~.
\end{align}
When $r_1 < r_2$ the additional component is $\{z_3=0\}$, so that
\begin{align}
F =  \{z_1=z_2=z_4=0\}\cup \{z_3=0\} \implies X = \P^2~.
\end{align}
The latter follows since when $z_3\neq 0$ we can set $z_3 = 1$ by the second $\C^\ast$ action.

The QSCR are
\begin{align}
(\sigma\cdot E_1) (\sigma\cdot E_2) (\sigma\cdot E_4) & = q_1~,&
(\sigma \cdot E_3) (\sigma \cdot E_4) & = q_2~.
\end{align}
One might naively think that the solution set would consist of $6$ points, but this is not the case.  Indeed, we can recast the system as
\begin{align}
\left[ q_2 (\sigma\cdot E_1) (\sigma\cdot E_2) - q_1 (\sigma \cdot E_3)\right] (\sigma\cdot E_4) & = 0~,
&
(\sigma \cdot E_3) (\sigma \cdot E_4) & = q_2~.
\end{align}
Since $\sigma \cdot E_4$ cannot vanish for any finite solution, we see that for generic values of the $(E,q)$ parameters, all solutions to the QSCR arise as solutions to the reduced system
\begin{align}
q_2 (\sigma\cdot E_1) (\sigma\cdot E_2) - q_1 (\sigma \cdot E_3) & = 0~,
&
(\sigma \cdot E_3) (\sigma \cdot E_4) & = q_2~.
\end{align}
Generically, there are $4$ solutions, consistent with $\chi(\F_1 ) = 4$.  For instance, for $\cE = T_X$ the system reduces to
\begin{align}
\sigma_2 &= q_1^{-1} q_2 \sigma_1~,& \sigma_1^4 + q_1 q_2^{-1} \sigma_1^3 &= q_1^2 q_2^{-1}~.
\end{align} 

The discriminant locus can be obtained by following the same compactification scheme that we used for the $\P^1\times\P^1$ example.  Introducing the projective coordinates $[s_0:s_1:s_2]$, we consider $\Vb \subset \P^2$ defined by
\begin{align}
q_2 (s\cdot E_1) (s\cdot E_2) - q_1 s_0(s \cdot E_3) & = 0~,
&
(s \cdot E_3) (s \cdot E_4) & = s_0^2q_2~.
\end{align}
We seek the parameter values for which $\Vb$ intersects the compactification divisor, which leads to the same form for the discriminant as for $X = \P^1\times\P^1$:
\begin{align}
D = \zeta_{13} \zeta_{14} \zeta_{23} \zeta_{24}~.
\end{align}
This is again in accord with intuition from the instanton expansion of the correlators: each correlator will be a polynomial in $q_2$ and $q_1 q_2^{-1}$, and the discriminant locus is expected to be $q$-independent.  The vanishing of $D$ corresponds to choices of $E$ for which $\cE$ is a sheaf and not a smooth vector bundle over $X$.

\subsubsection*{$\GL(2,\C)$ quotient and SUSY vacua}
Just as in the previous example, the reduced A/2 parameter space is $\Pi_{2,4}$, with the singular locus $\cA_{\text{sing}} = \{D = 0\} \subset \Gr(2,4)$.   In this case there are no continuous $\sigma$-branches, and the non-SUSY locus $\cA_{\nSUSY} \subset \cA_{\text{sing}}$ is closed and given by
\begin{align}
\cA_{\nSUSY} = \P^2_1 \cup \P^2_2 \cup \P^2_3\cup \P^2_4~,
\end{align}
where the $\P^2_i$ are defined as in~(\ref{eq:P2idef}).  This is shown in detail in appendix~\ref{app:nonSUSYdP1}.

\subsubsection*{The $T_4$ quotient and GIT}
A convenient basis for the $T_4$ action is
\begin{align}
\xymatrix@R=0mm@C=4mm{
~	&\zeta_{12}	&\zeta_{13}	&\zeta_{14}		&\zeta_{23}	&\zeta_{24}	&\zeta_{34}	&q_1		&q_2\\
t_1	&0			&1			&0			&1			&0			&1			&0		&1	\\
t_2	&0			&0			&1			&0			&1			&1			&1		&1	\\
%t_3	&0			&1			&1			&-1			&-1			&0			&0		&0 	\\
t_3	&1			&2			&0			&1			&-1			&0			&0		&0	\\
t_4	&1			&1			&-1			&2			&0			&0			&0		&0	\\
}
\end{align}
Unlike the previous case, here the $\GUL$ symmetry is completely broken, so that $T_4$ acts effectively.  The choice of basis shows that it factors through a $T_2$ action given by the last two rows, while the first two rows can be used to fix $q_1$ and $q_2$ to some fiducial values.  Including the $\C^\ast$ action on the $\zeta_{ij}$, we therefore consider the quotient by
\begin{align}
\xymatrix@R=0mm@C=4mm{
~		&\zeta_{12}	&\zeta_{13}	&\zeta_{14}	&\zeta_{23}	&\zeta_{24}	&\zeta_{34}	\\
\tau_0	&1			&1			&1			&1			&1			&1			\\
\tau_1	&1			&2			&0			&1			&-1			&0				\\
\tau_2	&1			&1			&-1			&2			&0			&0				\\
}
\end{align}
The secondary fan for this action is given in figure~\ref{fig:stabfordP1}.  
\begin{figure}[t]
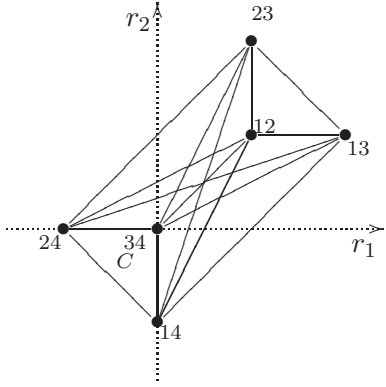

\[
\xy 
(0,0)*\xybox{<2.5mm,0.0mm>:
(0,0)*{\bullet}="34"; (-1.2,-0.7)*{{\scriptstyle{34}}};
(10,5)*{\bullet}="13"; (10.7,4.3)*{\scriptstyle{13}};
(-5,0)*{\bullet}="24"; (-5.7,-0.7)*{\scriptstyle{24}};
(5,10)*{\bullet}="23"; (5.6,11.5)*{\scriptstyle{23}};
(0,-5)*{\bullet}="14";(0.7,-5.5)*{\scriptstyle{14}};
(5,5)*{\bullet}="12";(5.7,5.5)*{\scriptstyle{12}};
(-1.7,-1.7)*{\scriptstyle{C}};
{\ar@{.>}(-8,0);(12,0)};
(11,-1)*{r_1};
{\ar@{.>}(0,-8);(0,12)};
(-1,11)*{r_2};
{\ar@{-}"34";"13"};{\ar@{-}"34";"14"};{\ar@{-}"34";"24"};{\ar@{-}"34";"23"};
{\ar@{-}"13";"14"};{\ar@{-}"14";"24"};{\ar@{-}"24";"23"};{\ar@{-}"23";"13"};
{\ar@{-}"12";"14"};{\ar@{-}"12";"13"};{\ar@{-}"12";"14"};{\ar@{-}"12";"23"};{\ar@{-}"12";"24"};{\ar@{-}"12";"34"};
{\ar@{-}"13";"24"};{\ar@{-}"14";"23"};
};
\endxy
\]
\caption{Stability conditions for the $T_4$ quotient of the $dP_1$ model: the secondary fan. The moment maps  for the $(\tau_0,\tau_1,\tau_2)$ are $r_0,r_1,r_2$, and shown is the intersection of the fan with the $r_0 = 1$ plane.}
\label{fig:stabfordP1}
\end{figure}
There are many choices of stability conditions, but, as in the previous example, all of them include some non-SUSY orbits.  This is easy to see by taking the image of the $\P^2_i$ under the moment map: this covers the entire secondary fan.

As a concrete example of a compactification of the A/2 moduli space, consider the full-dimensional cone $C$ in the figure.  This corresponds to the exceptional set
\begin{align}
F_C = \{\zeta_{14} = 0\} \cup \{\zeta_{24} = 0\} \cup \{\zeta_{12}=0,\zeta_{13}=0,\zeta_{23}=0,\zeta_{34}=0\}~.
\end{align}
The compactification can be described in terms of four $\tau_1,\tau_2$ invariants:   $\zeta_{34}$, as well as
\begin{align}
u &  = \zeta_{12} \zeta_{14} \zeta_{24}~,&
v & = \zeta_{13} \zeta_{14} \zeta_{24}^2~,&
w & = \zeta_{23} \zeta_{14}^2 \zeta_{24}~.
\end{align}
These satisfy $u\zeta_{34}-v+w = \zeta_{14}\zeta_{24} P$.  From the exceptional set $F_C$ we obtain
\begin{align}
\cMb(dP_1) = \{u\zeta_{34}-v+w = 0\} \subset \P^3_{3441} [u,v,w,\zeta_{34}]~.
\end{align}
The discriminant is $\Db = \{vw=0\}$; the non-SUSY locus is just a point:  $\overline{\cA_{\nSUSY} }= [1,0,0,0]$; the tangent bundle is the point $[0,1,1,-1]$.

\subsection{$X=\F_2$}
For our final example we consider a NEF Fano surface $X = \F_2$, to illustrate some important differences from the Fano case.  The charge matrix
\begin{align}
Q = \begin{pmatrix} 1 & 1 & 0 & 0 \\ -2 & 0 & 1 & 1 
\end{pmatrix}~
\end{align}
leads to the fans and phases in figure~\ref{fig:F2}.
\begin{figure}
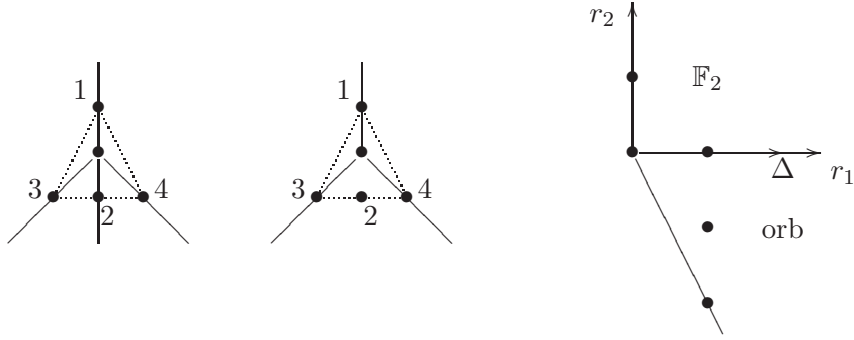

\[
\xy 
(-10,0)*\xybox{<1.2mm,0.0mm>:
(0,0)*{\bullet}="0";
(0,5)*{\bullet}="1"; (-2,7)*{1};
(-5,-5)*{\bullet}="3"; (-7,-4)*{3};
(0,-5)*{\bullet}="2"; (1,-7)*{2};
(5,-5)*{\bullet}="4";(7,-4)*{4};
{\ar@{-}"0";(0,10)};
{\ar@{-}"0";(-10,-10)};
{\ar@{-}"0";(0,-10)};
{\ar@{-}"0";(10,-10)};
{\ar@{.}"1";"3"};{\ar@{.}"3";"2"};{\ar@{.}"2";"4"};{\ar@{.}"4";"1"};
};
(25,0)*\xybox{<1.2mm,0.0mm>:
(0,0)*{\bullet}="0";
(0,5)*{\bullet}="1"; (-2,7)*{1};
(-5,-5)*{\bullet}="3"; (-7,-4)*{3};
(0,-5)*{\bullet}="2"; (1,-7)*{2};
(5,-5)*{\bullet}="4";(7,-4)*{4};
{\ar@{-}"0";(0,10)};
{\ar@{-}"0";(-10,-10)};
%{\ar@{-}"0";(0,-10)};
{\ar@{-}"0";(10,-10)};
{\ar@{.}"1";"3"};{\ar@{.}"3";"2"};{\ar@{.}"2";"4"};{\ar@{.}"4";"1"};
%%%
(30,0)*\xybox{<1.0mm,0.0mm>:
(0,0)*{\bullet}="0";
(10,0)*{\bullet}="1";
(10,10)*{\F_2};
(20,-10)*{ \text{orb}};
%(20,20)*{\bullet}="3";
(0,10)*{\bullet}="2";
(10,-10)*{\bullet};
(10,-20)*{\bullet};
{\ar@{->}"0";(25,0)};
{\ar@{-}"0";(12,-24)};
{\ar@{->}"0";(0,20)};
{\ar@{->}"0";(20,0)};
(20,-2)*{\Delta};
(28,-3)*{r_1};
(-4,18)*{r_2};
};
}
\endxy
\]
\caption{$\F_2$ fan, its alternate orbifold phase, and the secondary fan with $\Delta = -K_X$.}
\label{fig:F2}
\end{figure}
The QSCR are
\begin{align}
(\sigma \cdot E_1) (\sigma \cdot E_2) & =q_1~,&
(\sigma \cdot E_3) (\sigma \cdot E_4) & = q_2 (\sigma \cdot E_1)^2~.
\end{align}
In the second relation we brought the $(\sigma \cdot E_1)^2$ factor to the right-hand-side using the same logic as we did in reducing the system in the previous example:  for generic values of the parameters these polynomial equations describe the QSCR solution set. 

Our next task is to obtain the discriminant locus. By now the methodology is familiar:  we compactify the $\sigma$s to $\P^2$, and consider the intersection of the subvariety $\Vb$ defined by
\begin{align}
(s \cdot E_1) (s\cdot E_2) & =q_1 s_0^2~,&
(s \cdot E_3) (s \cdot E_4) & = q_2 (s \cdot E_1)^2~.
\end{align}
with the compactification divisor $C = \{s_0 = 0\}$.  The result is
\begin{align}
D = \zeta_{13}\zeta_{14} (\zeta_{23}\zeta_{24}  - q_2 \zeta_{12}^2)~.
\end{align}
The new feature is that now the discriminant depends on $q_2$.  Indeed, in this case correlators receive contributions from an infinite sum of instantons, and correlators develop $q$-dependent singularities due to diverging instanton sums~\cite{Morrison:1994fr}.  

\subsubsection*{GIT and SUSY vacua}
A convenient basis for the $T_4$ action on $\Pi_{2,4}$ is
\begin{align}
\xymatrix@R=0mm@C=4mm{
~	&\zeta_{12}	&\zeta_{13}	&\zeta_{14}	&\zeta_{23}	&\zeta_{24}	&\zeta_{34}	&q_1		&q_2\\
t_1	&0			&2			&2			&0			&2			&2			&0		&0	\\
t_2	&0			&1			&-1			&1			&-1			&0			&0		&0	\\
t_3	&1			&1			&1			&2			&2			&2			&1		&2	\\
t_4	&0			&0			&1			&0			&1			&1			&0		&1	\\
}
\end{align}
This A/2 theory has a $\Z_2$ global symmetry, and, as a consequence, there is an ineffective $\Z_2$ generated by $t_1 = -1$ and $t_{2,3,4} = 1$.  Focusing, as before, on the $T_2/\Z_2$ action, combined with the projective action on the $\zeta_{ij}$, we consider the quotient on $\Gr(2,4)$ generated by
\begin{align}
\xymatrix@R=0mm@C=4mm{
~		&\zeta_{12}	&\zeta_{13}	&\zeta_{14}	&\zeta_{23}	&\zeta_{24}	&\zeta_{34}	\\
\tau_0	&1			&1			&1			&1			&1			&1			\\
\tau_1	&0			&1			&1			&0			&0			&1				\\
\tau_2	&0			&1			&-1			&1			&-1			&0				\\
}
\end{align}
The stability conditions are encoded in the secondary fan in figure~\ref{fig:stabforF2}.
\begin{figure}[t]
\[
\xy 
(0,0)*\xybox{<3.0mm,0.0mm>:
(5,0)*{\bullet}="34"; (5.7,0.4)*{{\scriptstyle{34}}};
(5,5)*{\bullet}="13"; (5.7,5.5)*{\scriptstyle{13}};
(0,-5)*{\bullet}="24"; (-0.7,-5.1)*{\scriptstyle{24}};
(0,5)*{\bullet}="23"; (-0.7,5.1)*{\scriptstyle{23}};
(5,-5)*{\bullet}="14";(5.7,-5.5)*{\scriptstyle{14}};
(0,0)*{\bullet}="12";(-0.7,0.6)*{\scriptstyle{12}};
(4.4,2.3)*{\scriptstyle{C}};
{\ar@{.>}(-8,0);(12,0)};
(11,-1)*{r_1};
{\ar@{.>}(0,-7);(0,10)};
(-1,9)*{r_2};
{\ar@{-}"34";"13"};{\ar@{-}"34";"14"};{\ar@{-}"34";"24"};{\ar@{-}"34";"23"};
{\ar@{-}"13";"14"};{\ar@{-}"14";"24"};{\ar@{-}"24";"23"};{\ar@{-}"23";"13"};
{\ar@{-}"12";"14"};{\ar@{-}"12";"13"};{\ar@{-}"12";"14"};{\ar@{-}"12";"23"};{\ar@{-}"12";"24"};{\ar@{-}"12";"34"};
{\ar@{-}"13";"24"};{\ar@{-}"14";"23"};
};
\endxy
\]
\caption{Stability conditions for the $T_4$ quotient of the $\F_2$ model: the secondary fan. The moment maps  for the $(\tau_0,\tau_1,\tau_2)$ are $r_0,r_1,r_2$, and shown is the intersection of the fan with the $r_0 = 1$ plane.}
\label{fig:stabforF2}
\end{figure}
For instance, taking the full-dimensional cone $C$ leads to the exceptional set
\begin{align}
F_C = \{\zeta_{13} =0\}\cup \{\zeta_{14} =0,\zeta_{34} =0\} \cup \{\zeta_{12}=0,\zeta_{23}=0,\zeta_{24} = 0\} \subset \{D=0\}~.
\end{align}
So, we obtain a compactification of the A/2 moduli space
\begin{align}
\cMb_C(\F_2) = \{P = 0\} \subset Y~,
\end{align}
where $Y$ is the toric variety $Y = \{\C^6 \setminus F_C\} / (\C^\ast)^3$.  The discriminant locus is obtained by intersecting this further with
\begin{align}
\Db = \zeta_{14}(\zeta_{23}\zeta_{24}-q_2\zeta_{12}^2)~.
\end{align}
In appendix~\ref{app:nonSUSYF2} we show that the non-SUSY locus is given by
\begin{align}
\overline{\cA_{\nSUSY}} = \{\Delta_1 = 0,~\zeta_{14} = 0\} \cup \{\Delta_1 = 0,~\Delta_2 = 0\}~,
\end{align}
where
\begin{align}
\Delta_1 &= \zeta_{23}\zeta_{34} - 2 q_2 \zeta_{12} \zeta_{13}~,&
\Delta_2 &=\zeta_{13}\zeta_{24}+\zeta_{14}\zeta_{23}~.
\end{align}

\section{Discussion} \label{s:discuss}
In this work we described the parameter space of the A/2 half-twisted theory for a NEF Fano GLSM with toric deformations of the tangent sheaf.  The key element in the analysis was the identification of various interesting loci in the $\cP = \{(E,q)\}$ parameter space via the spectral cover associated to the quantum sheaf cohomology relations.  We identified singular theories, as well as those with spontaneously broken supersymmetry.  We showed that the space of all such theories modulo field redefinitions is described by the quantum moduli stack $\Pi_{k,n}/T_k$.  In examples, we were able to show that there exist GIT stability conditions that lead to separated moduli spaces of A/2 theories.  The most pressing step in completing our analysis of this class of models is to characterize such stability conditions for smooth NEF Fano $X$.   In addition, there are many interesting questions about the algebraic geometry of these moduli spaces:  what algebraic varieties arise as NEF Fano A/2 moduli spaces?  what sorts of singularities do they exhibit? do they, for instance, satisfy ``Murphy's law'' of moduli spaces~\cite{Payne:2008tb}?  what is the algebraic geometry of the non-SUSY locus?  

There are a number of natural generalizations of our analysis.  For instance, it should be possible to consider more general GLSMs for simplicial and projective toric varieties.  In this context the weakly coupled UV phase of the theory will in general be a NLSM with some toric target-space $X_{\text{UV}}$, and it should be possible to relate the spectral cover for the QSCR to properties of $X_{\text{UV}}$.  Typically $X_{\text{UV}}$ will not be smooth, and we expect to find appropriate generalizations of various geometric concepts.  For instance, the number of solutions to the QSCR should be related to some stringy generalization of the Euler characteristic. 

Another generalization would be to study the most general monadic deformations of $X$.  This will require understanding the quotient of the GLSM parameters by the full group of automorphisms of the toric variety.  While this has an explicit description~\cite{MR1299003}, the group is not in general reductive.  Of course there will be plenty of examples where the group will be reductive, and at least for those we would expect to find similarly nice moduli spaces.

Perhaps a more interesting direction would be to study massive theories with gauge-neutral $\Sigma$ multiplets but no (2,2) locus.  In other words, we could start with a more general toric monadic bundle and study its A/2 twist and, correspondingly, the IR dynamics of the ground ring.  To our knowledge such theories have not been explored in any detail.

While much remains to be learned about the massive theories, the most interesting step would be to go beyond these ``model'' problems and tackle the one that provided our original motivation:  the moduli space of (0,2) SCFTs based on GLSM constructions.  The ideas developed in this work, together with results like the (0,2) quantum restriction formula~\cite{McOrist:2008ji} should allow us to study monadic moduli spaces of at least some simple examples.  Understanding appropriate stability conditions will be a key challenge in this generalization.  For the massive A/2 theories considered in this work, the choice of exceptional set $F$, a choice of a GIT stability condition, is not fixed by any requirement beyond the desire to keep all smooth points.  We expect that in a (0,2) SCFT, corresponding to a Calabi-Yau geometry, a stability condition should be singled out by the (0,2) quantum generalization of Hermitian Yang-Mills and corresponding $0$-slope stability familiar from supergravity.  It should be very interesting to see whether the GLSM compactification of the (0,2) SCFT moduli space includes limit points with spontaneously broken SUSY.  

\appendix

\section{Remarks on non-linear $E$ deformations} \label{app:nonlin}

The arguments presented here were sketched out by IVM and M.R.~Plesser on a drive from Duke to Virginia Tech on a blustery October day in 2011.  IVM thanks MRP for the ride and for his contribution in making this argument.

Consider a toric GLSM as in the main body of the text with an arbitrary $\bE$ deformation.  To characterize these, we first consider the monomials in $S=\C[z_{\rho_1},\ldots,z_{\rho_n}]$ that have the same gauge charges as a $z_\rho$.  That is, we define\footnote{This is slightly different from a similar definition in~\cite{Kreuzer:2010ph}, where $S_\rho$ also included $z_\rho$.  The sets $S_\rho$ play an important role in characterizing the automorphisms of a toric variety~\cite{MR1299003,Cox:2000vi}.}
\begin{align}
S_{\rho} = \{ \mon \in S ~~|~~ \mon\neq z_{\rho'} \quad\text{and}\quad  \deg \mon =\deg z_\rho \}~.
\end{align}
This allows us to write the full E-couplings as
\begin{align}
\cDb \Gamma_\rho = \sum_{a=1}^{k} \Sigma_a \left[ \sum_{\rho'} E^a_{\rho\rho'} z_{\rho'} + \sum_{\mon \in S_\rho} F^a_{\rho \mon} \mon \right]~.
\end{align}
Of course $E^{a}_{\rho\rho'}$ is zero whenever $\deg z_\rho \neq \deg z_{\rho'}$.
As remarked in the text, we can instead work with an isomorphic GLSM where we relabel $\Gamma_\rho$ and $\Gammab_\rho$ (we will use $\Gammat$ to denote these relabeled multiplets) and work with chiral $\Gammat$ and Yukawa couplings encoded in a (0,2) superpotential $\cL \supset \int d\theta \cW + \text{h.c.}$ and
\begin{align}
\cW = 
\sum_{a=1}^{k}\left[  \frac{\log q_a}{8\pi i}  \Upsilon_a + 
\sum_{\rho} \Gammat^\rho \left( \sum_{\rho'} E^a_{\rho\rho'} Z_{\rho'} + \sum_{\mon \in S_\rho}  F^a_{\rho\mon} \mon \right) \Sigma_a
\right]~.
\end{align}
Note that the gauge charges of $\Gammat^\rho$ are opposite to those of $\Gamma_\rho$.

The task before us is to constrain possible quantum corrections to this superpotential.  This can be accomplished by using the familiar tools of holomorphy~\cite{Seiberg:1994bp,Intriligator:1994jr,Argyres:book}.   Setting $F=0$, we observe that the theory preserves a large non-anomalous global symmetry $\GU(1)_{\text{tot}}$ that assigns charge $+1$ to $\Gammat^\rho$ and $-1$ to $Z_\rho$.  Turning on the $F$ couplings will break this symmetry, but we can still obtain selection rules by assigning charges to the couplings:
\begin{align}
\xymatrix@R=2mm@C=2mm{
~		&\Gammat^\rho		&Z_\rho	&E^a_{\rho\rho'}	& \mon = \prod_\rho Z_\rho	&F^{a}_{\rho\mon}	&q_a~~\\
\GU(1)_{\text{tot}}	&+1				&-1		&0				&-\ell(\mon)				&+\ell(\mon)		&0~,
}
\end{align}
where $\ell(\mon) = \sum_\rho m_\rho$. By assumption $\ell(\mon) > 1$.

When $F =0$ and $E$ is full rank, then large $\sigma_a$ expectation values mass up the $(\Gammat^\rho,Z_\rho)$, and integrating these out at one loop leads to an effective potential~\cite{McOrist:2007kp}  
\begin{align}
\cW_{\text{eff}}^0= \frac{1}{8\pi i} \int d\theta \sum_a \Upsilon^a \Jt_a(q,E,\Sigma) + \text{h.c.}~.
\end{align}
%\begin{align}
%\cW_{\text{eff}}^{0} & =\sum_{a=1}^{k}  \frac{1}{8\pi i} \log\left[ 
%\frac{q_a}{\prod_\rho (\mu^{-1} \sum_{b=1}^{r} \Sigma_b E^b_\rho)^{Q^a_\rho}} \right] \Upsilon_a~.
%\end{align}
How is this modified if $F^a\neq 0$?  Suppose the $E^a_\rho$ are such that they lead to a smooth A/2 theory.  We do not expect this to be qualitatively modified by turning on small $F$ couplings; hence any correction to $\cW^0_{\text{eff}}$ should be a holomorphic function of the $F^a$; moreover, the superpotential must be linear in the fermi fields $\Upsilon_a$, i.e. we can write the full effective potential as
\begin{align}
\cW_{\text{eff}} = \cW^0_{\text{eff}} +  \sum_{a=1}^k \cJ_a (F,E,q,\Sigma)\Upsilon^a~.
\end{align}
This must be invariant under $\GU(1)_{\text{tot}}$, but any holomorphic dependence of $\cJ_a$ on $F$ necessarily carries a positive $\GU(1)_{\text{tot}}$ charge.  Hence, $\cJ_a = 0$, and the effective potential  is simply $\cW^0_{\text{eff}}$.
 This is consistent with an instanton-by-instanton analysis presented in appendix A of~\cite{Donagi:2011va}, where it is shown that non-linear $E$ deformations do not affect A/2 correlators for a compact  toric variety $X$.

A few additional comments on this result may be useful.  First, when $X$ is non-compact we expect Higgs vacua in the IR as well, so this potential will not describe the full IR dynamics.  This is of course an issue already in (2,2) theories; it has been explored in some detail in~\cite{Melnikov:2006kb}, but there are many interesting questions that still remain about the interplay of the different vacua.  Getting back to compact $X$, when $E^a_\rho$ is not full rank, then, as discussed in the main body of the text, $\cW_{\text{eff}}$ predicts two IR phenomena:  the theory will not have a gap and (0,2) SUSY will be spontaneously broken.  
It is also easy to see that typically the theory will not have a simple relationship to a (0,2) NLSM for $\cE \to X$.  A very simple example of this is a $\C\P^1[1,N]$ (2,2) GLSM deformed so that the E-couplings are of the form $\Gammat^2 \Sigma Z_1^N$.  The UV NLSM has a singularity at $Z_1 = 0$, where a $\Sigma$ branch emerges.  It should be useful to study such theories in more detail.  We suspect that (0,2) SUSY will still be broken, but what can we say about the light degrees of freedom?  

A final caveat concerns any applications to theories that should flow to compact SCFTs.  These applications, while possible~\cite{Morrison:1994fr,McOrist:2008ji}, should be made with care.  The models will typically have both $E$ and $J$ couplings and no massive $\Sigma$ vacua.

\section{Some elliptic genera} \label{app:elgen}
In this appendix we will check that for the $\P^1$ model the elliptic genus is consistent with the patterns of SUSY breaking indicated by the effective potential.  There are three cases to consider:  $E = (\ast,\ast)$, $E=(\ast,0)$, and $E=(0,0)$; here $\ast$ indicates a non-zero entry.  The first and last have been discussed before, for instance in~\cite{Gadde:2013lxa}, but we include them for completeness and comparison.

The flavored elliptic genus of a general (0,2) GLSM is computed by a residue formula described in~\cite{Gadde:2013dda,Benini:2013xpa}.  We will present here just the few details we need for our computation, and the reader is encouraged to consult the references for the full story.

\subsection{A simple set-up}
Suppose we have a $\GU(1)$ (2,2) toric GLSM with gauge charges $Q_\rho >0$ for all (2,2) chiral multiplets.  We deform this to a more general (0,2) theory by E-deformations that preserve a global symmetry group $G$ with maximal torus $T_G \subset G_{\C}$.  We wish to compute the flavored elliptic genus, in other words the Ramond-Ramond $T^2$ partition function $Z(q,y)$, where $q = e^{2\pi i\tau}$ keeps track of the torus complex structure and $y_s = e^{2\pi i z_s}$, $s=1,\ldots, \dim T_G$, are the fugacities that keep track of the characters with respect to $T_{G}$.  In the path integral the $z_s$ should be thought of as holonomies for background gauge fields valued in $T_G$.  The flavored elliptic genus is then given as a residue
\begin{align}
Z = \sum_{v \in M} \oint_{u = v} {\!\!du}~ \frac{\eta(q)^3}{i}  \frac{1}{\theta_1(q,y^{R(\Sigma)})} \prod_{\rho} \frac{ -\theta_1(q,e^{2\pi i u Q_\rho} y^{R(\Gamma^\rho)})}{\theta_1(q,e^{2\pi i u Q_\rho} y^{R({Z_\rho})})} .
\end{align}
Here $u$ denotes the holonomy of the background gauge field, which is integrated over, and $y^R = \prod_s y_s^{R_s}$ are the characters of the indicated multiplets, and $\theta_1(q,y)$ and $\eta(q)$ are the Jacobi and Dedekind functions
\begin{align}
\theta_1 (q,y) &=  -i q^{1/8} y^{1/2} \prod_{k=1}^\infty(1-q^k) (1-yq^k) (1-y^{-1} q^{k-1})~,&
\eta(q) &= q^{1/24} \prod_{n=1}^{\infty} (1-q^n)~.
\end{align}
The integrand is a product of factors associated to the vector, chiral, and fermi multiplets.  When viewed as a function of $u$ the integrand has poles when the $\theta_1$ terms in the denominator develop simple zeroes.  The most interesting and subtle part of the computation is to properly identify the integration contour.  Fortunately, in the case at hand this is very simple:  $M$ includes all poles (modulo $u \sim u+ \Z +\tau\Z$) where $e^{2\pi i Q_\rho u} y^{R({Z_\rho})} = 1$ for some $\rho$.  The residues are easily evaluated by using $\left.\p_u \theta_1(q,e^{2\pi i u})\right|_{u=0} = 2\pi \eta(q)^3$.

Thus, in any particular example all we need to do is to determine the global symmetries and sum up the residues of the integrand.  In doing so, there is one more subtlety:  in general the global symmetries will be anomalous, which will translate into the integrand not being consistent with $u \sim u + \Z+\tau \Z$.  This requires us to restrict the fugacities so that
\begin{align}
\prod_\rho y^{Q_\rho R({\Gamma^\rho}) -  Q_\rho R({Z_\rho})} = 1.
\end{align}
Having reviewed this basic technology, we will now apply it to the simple example of $\P^1$, where $\rho \in \{1,2\}$, and $Q_\rho =1$.

\subsection{$E = (\ast,\ast)$}
We begin with the generic point, which includes the (2,2) locus at $E = (1,1)$.  The couplings preserve a rank $3$ symmetry with the following charge assignments for $T_G$.
\begin{align}
\xymatrix@R=2mm@C=5mm{
\text{symmetry}		
	& Z_1	&Z_2	&\Gamma^1	&\Gamma^2	&\Sigma	&\text{fugacity}\\
\GU(1)_{\text{gauge}}
	&1		&1		&1			&1			&0		&x  \\
\GUL& 0		&0		&-1			&-1			&-1		&y_1\\
\GU(1)_{2}
	&1		&-1		&1			&-1			&0		&y_2
}
\end{align}
As a result, the integrand takes the form
\begin{align}
%%%%
I &= -\eta(q)^2 \frac{i\eta(q) }{\theta_1(q,y_1^{-1})} 
%%%%
\frac{ i \eta(q)}{\theta_1(q,xy_2)} 
\frac{ i \eta(q)}{\theta_1(q,x y_2^{-1})} 
%%%%
\frac{ i\theta_1(q,xy_1^{-1}y_2)}{\eta(q)}
\frac{ i\theta_1(q,xy_1^{-1}y_2^{-1})}{\eta(q)} \nonumber\\[2mm]
&= \frac{\eta(q)^3}{i} \frac{\theta_1(q,xy_1^{-1}y_2)\theta_1(q,xy_1^{-1}y_2^{-1})}{\theta_1(q,y_1^{-1})\theta_1(q,xy_2)\theta_1(q,x y_2^{-1})}
\end{align}
This is gauge-invariant if and only if $y_1^2 =1$.  Summing the residues at $u = \pm z_2$ leads to
\begin{align}
Z = \frac{\theta_1(q,y_1^{-1} y_2^{-2})}{\theta_1(q, y_2^{-2})} + \frac{\theta_1(q,y_1^{-1} y_2^{2})}{\theta_1(q, y_2^{2})} =
2 \frac{ \theta_1(q,y_1^{-1} y_2^{2})}{\theta_1(q, y_2^{2})}~.
\end{align}
Setting $y_1 = 1$, we obtain the correct Witten index:  $Z = 2$;  setting $q \to 0$ and using
\begin{align}
\lim_{q\to 0} \frac{\theta_1(q,z)}{\theta_1(q,w)} = \frac{ 1-z^{-1}}{1-w^{-1}}~,
\end{align}
we obtain $ \lim_{q\to 0} Z = (1+y_1)$.  This leading behavior is reproduced by the effective field theory based on the $J(\Sigma)$ superpotential. 

\subsection{$E = (0,0)$}
Next, we consider the case of a vanishing $E$.  In this case there is a rank $5$ symmetry with the following charges.
\begin{align}
\xymatrix@R=2mm@C=5mm{
\text{symmetry}		&Z_1	&Z_2	&\Gamma^1	&\Gamma^2	&\Sigma	&\text{fugacity}\\
\GU(1)_{\text{gauge}}
	&1		&1		&1			&1			&0		&x  \\
\GUL& 0		&0		&-1			&-1			&-1		&y_1\\
\GU(1)_{2}
	&1		&-1		&0			&0			&0		&y_2 \\
\GU(1)_{3}
	&0		&0		&1			&-1			&0		&y_3 \\
\GU(1)_{4}
	&0		&0		&0			&0			&1		&y_4
}
\end{align}
The integrand is
\begin{align}
I = \frac{\eta(q)^3}{i} \frac{ \theta_1(q,xy_1^{-1} y_3) \theta_1(q,xy_1^{-1} y_3^{-1})}{\theta_1 (q,y_1^{-1} y_4) \theta_1(q,xy_2) \theta_1(q,xy_2^{-1})}~;
\end{align}
as before it is gauge-invariant if and only if $y_1^2 =1$.  Evaluating the residues, which are still just at $xy_2 =1$ and $xy_2^{-1} =1$, we obtain
\begin{align}
Z = \frac{1}{\theta_1(q,y_1^{-1} y_4)}\left[\frac{\theta_1(q,y_1^{-1} y_2^{-1} y_3) \theta_1(q,y_1^{-1} y_2^{-1} y_3^{-1})}{\theta_1(q,y_2^{-2})} +
\frac{\theta_1(q,y_1^{-1} y_2 y_3)\theta_1(q,y_1^{-1} y_2 y_3^{-1})}{ \theta_1(q,y_2^2)}\right]~.
\end{align}
Imposing $y_1^2 =1$ and using $\theta_1(q,y) = -\theta_1(q, y^{-1})$ on the first term in the bracket, we get
\begin{align}
Z =  \frac{\theta_1(q,y_1^{-1} y_2 y_3)\theta_1(q,y_1^{-1} y_2 y_3^{-1})}{\theta_1(q,y_1^{-1} y_4) \theta_1(q,y_2^2)} (-1 +1) = 0.
\end{align}
The index vanishes.  This is consistent with the claim that SUSY is broken.

\subsection{$E=(\ast,0)$}
The previous two cases have already been considered in the literature; this one is new.
 The couplings $\cDb \Gamma^1 = \Phi^1 \Sigma$ and $\cDb \Gamma^2 = 0$ allow a rank $4$ symmetry:
\begin{align}
\xymatrix@R=2mm@C=5mm{
\text{symmetry}		& Z_1	&Z_2	&\Gamma^1	&\Gamma^2	&\Sigma	&\text{fugacity}\\
\GU(1)_{\text{gauge}}
	&1		&1		&1			&1			&0		&x  \\
\GUL& 0		&0		&-1			&-1			&-1		&y_1\\
\GU(1)_{2} 
	&1		&-1		&1			&-1			&0		&y_2 \\
\GU(1)_{3}
	&1		&-1		&0			&0			&-1		&y_3 \\
}
\end{align}
As before, $\GUL$ is the only symmetry with an anomaly, and that again restricts $y_1^2 =1$.
The integrand is
\begin{align}
I = \frac{\eta(q)^3}{i} \frac{ \theta_1(q,xy_1^{-1} y_2) \theta_1(q, x y_1^{-1} y_2^{-1})}{\theta_1(q,y_1^{-1} y_3^{-1}) \theta_1(q, xy_2y_3) \theta_1(q,xy_2^{-1} y_3^{-1})}~.
\end{align}
Taking the residues, we obtain
\begin{align}
Z = \frac{1}{\theta_1(q,y_1^{-1} y_3^{-1})} 
\left[ \frac{\theta_1(q, y_1^{-1} y_3^{-1}) \theta_1(q, y_1^{-1} y_2^{-2} y_3^{-1})}{\theta(q,y_2^{-2} y_3^{-2})}
+ \frac{ \theta_1(q,y_1^{-1} y_2^2 y_3) \theta_1(q, y_1^{-1} y_3)}{\theta(q, y_2^2 y_3^2)}
\right]
\end{align}
This vanishes for the same reason as in the previous section.

\subsection{Parting comments}
It would be interesting to compute the flavored elliptic genus for more examples to explore SUSY--breaking loci further.  For instance, it might be interesting to consider cases where SUSY vacua only exist for special values of the K\"ahler parameters $q_a$ (not to be confused with $q= e^{2\pi i\tau}$ used throughout this appendix),  as in the $\P^1\times \P^1$ theory, where
\begin{align}
E = \begin{pmatrix} 1 & 0 & 1 & 0 \\ 0 &1 & 0 &1
\end{pmatrix}
\end{align}
leads to SUSY vacua if and only if $q_1 = q_2$.  For generic $q_a$, therefore, we expect that the index will vanish.  To find a non-zero index for $q_1 = q_2$ one will probably need to introduce the additional discrete $\Z_2$ symmetry on that locus, which acts by exchanging the multiplets associated to the $\P^1$ factors:
\begin{align}
(Z_1,Z_2,\Gamma^1,\Gamma^2, V_1;Z_3,Z_4,\Gamma^3,\Gamma^4, V_2) \mapsto
(Z_3,Z_4,\Gamma^3,\Gamma^4, V_2;Z_1,Z_2,\Gamma^1,\Gamma^2, V_1)~.
\end{align}

\section{Solutions to the QSCR and $\chi(X)$} \label{app:bernstein}
The main goal of this appendix is to prove the following theorem; in addition, we give a short review of the combinatorics that determine the discriminant locus of the QSCR for any NEF Fano $X$.

\begin{thm}\label{thm:QSCREuler} The degree of the spectral cover of the GLSM of a smooth NEF Fano toric variety $X$  is given by the Euler characteristic of $X$.
\end{thm}
Note that the degree of the spectral cover is just the number of solutions $(\sigma_1,\ldots,\sigma_k) \in \C^k$ to the QSCR with generic $E$ and $q$ parameters.

The first step in the proof is to recast the QSCR of (\ref{eq:QSCR}) in terms of the Gale dual~\cite{Cox:2011tv} of the $E$ matrix, $\Eh$.  For our purposes, all we need to know about $\Eh$ is that it is a $d\times n$ matrix such that solutions to $\Eh \cdot z = 0$ are given by $z_\rho = \sum_a \sigma_a E^a_\rho $.  With this, we recast the QSCR as the system
\begin{align}
\prod_\rho z_\rho^{Q^a_\rho} &= q_a~,&
\Eh \cdot z & = 0~.
\end{align}
As long as the A/2 theory is non-singular $\sigma_a$ for every solution must be finite, and this in turn requires the $z_\rho$ to be finite for every solution.  Since $X$ is a compact variety with a pointed secondary fan, it follows that we can choose a basis for the charges $Q^a_\rho$ such that $Q^1_\rho > 0$ for all $\rho$.  Thus, as long as $q_1 \in \C^\ast$, every $z_\rho$ is non-vanishing.  So, for generic parameters every solution corresponds to $z \in (\C^\ast)^n$, and these can be counted by applying the well-known theorem of Bernstein~\cite{Bernstein:1975nr,Fulton:1993tv,Sturmfels:2002sp}.
\begin{thm}[Bernstein]
Let $( f_1,\ldots, f_n) \in \C[z_1^{\pm1}, z_2^{\pm1},\ldots,z_n^{\pm1}]$ be a system of $n$ Laurent polynomials with supports $(S_1,\ldots, S_n)$ in lattice  $\Mt\simeq \Z^{n}$ and generic coefficients.\footnote{The lattice $\Mt$ should not be confused with the $M\simeq \Z^d$ lattice associated to the toric variety $X$.}  The number of solutions $z \in (\C^\ast)^n$ is then counted by the mixed volume $\MV_n(S_1,\ldots,S_n)$.\footnote{ For a pedagogical description of mixed volumes, we refer to~\cite{Cox:1998ua,Steffens:2009mixed}.  The general definition is as follows:  given convex bodies $K_1$, \ldots, $K_l$ in $\R^l$, we compute the volume of the Minkowski sum $\lambda_1 K_1+ \cdots+\lambda_l K_l$ for real non-negative parameters $\lambda$.  This is a homogeneous degree $l$ polynomial in the $\lambda$, and  the mixed volume $\MV_l(K_1,\ldots, K_l)$ is the coefficient of $\lambda_1\lambda_2\cdots\lambda_l$.}  Moreover, there is an explicit combinatorial description of the non-generic locus.
\end{thm}
We now apply the theorem to our specific system, which has the supports
\begin{align}
S_a &= \{(0,\ldots,0), (Q^a_1,Q^a_2,\ldots,Q^a_n) \}~, & a&=1,\ldots, k~,\nonumber\\
S_{k+1} & =S_{k+2} = \cdots = S_{k+d} = \{e_1,e_2,\ldots, e_n\} ~,
\end{align}
where the $e_\rho$ denote the standard basis for $\Mt$.

Since $X$ is projective and simplicial, we have the exact sequences( see theorem 6.4.1 of~\cite{Cox:2011tv} for more details)
\begin{align}
\label{eq:2ses}
\xymatrix{ 
0 \ar[r] & M \ar[r]^-{\alpha} & L \ar[r]^-{\beta} & \Pic(X) \ar[r] & 0~~,~ \\
%0 \ar[r] & M_{\R} \ar[r]^-{\alpha} & \R^{\Sigma_X(1)} \ar[r]^-{\beta} & \Pic(X)_\R \ar[r] & 0~~,~ \\
0 \ar[r] & N_1(X) \ar[r]^-{\beta^\ast} &L^\ast \ar[r]^-{\alpha^\ast} & N \ar[r] & 0 ~~,
}
\end{align}
where $N_1(X) = \Pic(X)^\ast$ is the free abelian group generated by complete irreducible curves in $X$ modulo numerical equivalence, $L = \Z^{\Sigma_X(1)}$, and the maps are given by
\begin{align}
\alpha^\ast(e_\rho) &= u_\rho~,&
\beta^\ast([C]) & = D_\rho \cdot C~,
\end{align}
where $u_\rho$ denotes the primitive lattice vector on the ray $\rho$. 
More concretely, the maps $\alpha$ and $\beta$ can be represented by an $n\times d$ matrix $A$ and an $k\times n$ matrix $Q$.  For instance, for $X = \P^1\times \P^1$, we have
\begin{align}
A^T &= \begin{pmatrix} 1 & -1 & 0 & 0 \\ 0 & 0 & 1 & -1 \end{pmatrix}~,&
Q & = \begin{pmatrix} 1 & 1 & 0 & 0 \\ 0 & 0 & 1 & 1 \end{pmatrix}~.
\end{align}
We now obtain the following lemma.
\begin{lem}\label{lem:detdet} If $X$ is smooth and projective, then $\det A^T A = \det Q Q^T$.
\end{lem}
\begin{proof}
Since $X$ is smooth and projective, $\Pic(X) \simeq \Z^k$.  Choose a unimodular definite pairing on $L$.  This yields an isomorphism $\vphi : L \simeq L^\ast$ and therefore injective maps
\begin{align}
\alpha^\ast\vphi\alpha &: M \to N~,&
\beta \vphi^{-1} \beta^\ast &:  N_1(X) \to \Pic(X)~.
\end{align}
Moreover, these maps have equal indices,  $[N : M] = [\Pic(X) : N_1(X)]$, since
\begin{align}
[\Pic(X) :N_1(X)] = [L/M : N_1(X)] = [L : M+N_1(X)] = [L^\ast/N_1(X): M] = [N: M]~.
\end{align}
But, since $\vphi$ is unimodular, this is equivalent to $\det A^T A = \det Q Q^T$.

%Thus, the first line of~(\ref{eq:3ses}) implies the Smith normal forms
%\begin{align}
%A &= P_1 \begin{pmatrix} \iden_d \\ 0 \end{pmatrix} P_2~,&
%Q &= P_3 \begin{pmatrix} 0 & \iden_k \end{pmatrix} P_1^{-1}~,
%\end{align}
%where $P_1 \in \SL(n,\Z)$, $P_2 \in \SL(d,\Z)$, and $P_3 \in \SL(k,\Z)$.  Without loss of generality we can set $P_2 =\iden_d$ and $P_3 = \iden_k$ by a choice of basis on $M$ and on $\Pic(X)$.  

%Since $A$ and $Q$ are full rank it is clear that $A^T A$ and $QQ^T$ are invertible.  Next, we decompose $P_1^T P_1$ as
%\begin{align}
%P_1^T P_1 = \begin{pmatrix} L_1 & L_2 \\ L_2^T & L_3 \end{pmatrix}~,
%\end{align}
%where $L_1$ and $L_3$ are, respectively $d\times d$ and $k\times k$ blocks.  Since $A^T A = L_1$, it follows that $L_1$ is invertible, so that 
%\begin{align}
%P_1^T P_1 = \begin{pmatrix} L_1 & 0 \\ L_2^T  & \iden_k \end{pmatrix} \begin{pmatrix} \iden_d & L_1^{-1} L_2 \\ 0 &  L_3-L_2^T L_1^{-1} L_2 \end{pmatrix}~.
%\end{align}
%Taking determinants of both sides, we obtain
%\begin{align}
%1 = \det L_1 \det (L_3 - L_2^T L_1^{-1} L_2)~,
%\end{align}
%so that $L_4 = L_3 - L_2^T L_1^{-1} L_2$ must also be invertible.  Computing the inverse $(P_1^T P_1)^{-1}$, we find $QQ^T = L_4^{-1}$, and the assertion follows.
\end{proof}

Using the unimodular pairing on $L$ we can decompose
%
%Next, we tensor~(\ref{eq:2ses}) with $\R$ and observe that $Q^T Q$ and $A A^T$ are two commuting positive semi-definite operators on $\R^{\Sigma_X(1)} = L \otimes_{\Z} \R$ with ranks, respectively, $k$ and $d$ and complementary kernels.  Hence, we can decompose
\begin{align}
L = \beta^\ast (N_1(X)) \oplus \left[  \beta^\ast (N_1(X))\right]^{\perp}~,
\end{align}
with the second factor isomorphic to $N_{\R}$ under the restriction of $\alpha^\ast$. This basic linear algebra allows us to apply a ``separation lemma'' to our mixed volume computation:
\begin{lem}[Lemma 2.6 of~\cite{Steffens:2009mixed}] Let $S_1,\ldots, S_k$ be polytopes in $\R^{k} \subset\R^{k+d}$ and $S_{k+1},\ldots, S_{k+d}$ be polytopes in $\R^n$.  Then
\begin{align*}
\MV_n(S_1,\ldots,S_n) = \MV_k(S_1,\ldots, S_k) \MV_{d} (\pi(S_{k+1}),\ldots, \pi(S_{k+d}))~, 
\end{align*}
where $\pi$ is the projection $\pi : \R^{k+d} \to \R^d$.
\end{lem}

We now use basic properties of mixed volumes to evaluate the two factors.  We begin with $\MV_k(S_1,\ldots, S_k)$.  Each $S_a$ is a one-dimensional polytope consisting of the origin and the vector $Q_a \in \R^n$.  Since these vectors are linearly independent, we have $\MV_k(S_1,\ldots,S_k) = \Vol_k (\cP)$, where $\cP$ is the parallelotope generated by the $k$ vectors $Q_a$.\footnote{Concretely, $\cP =\{ \sum_a t_a Q_a ~~|~~0\le t_a \le 1\} \subset\R^k \subset \R^n$.}  To compute this parallelotope volume, we note that $\Vol_k (\cP) = \Vol_n (\cP')$, where
\begin{align}
\cP' = \{ \textstyle\sum_a t_a Q_a + \sum_{\alpha=1}^d s_\alpha U_\alpha ~~|~~ 0 \le t_a,s_\alpha \le 1\}~,
\end{align} 
where $\{U_\alpha\}$ is an orthonormal basis for $\beta^\ast(N_1(X))^\perp$.  But then
\begin{align}
\Vol_n (\cP') =  \left| \det \begin{pmatrix} Q_1 &Q_2& \cdots& Q_k& U_1 &U_2 &\cdots& U_d \end{pmatrix} \right|,
\end{align}
whence
\begin{align}
\Vol_k (\cP)^2 = \Vol_n(\cP')^2 = \det Q Q^T~.
\end{align}
Next, we turn to the second factor and use a basic property of mixed volumes:
\begin{align}
\MV_d(\underbrace{P, \ldots, P}_{\text{$d$ times}} ) = d! \Vol_d (P)~.
\end{align}
Thus, since $S_{k+1}=S_{k+2}=\cdots=S_{n} = \Conv(e_1,\ldots,e_n)$, and $\alpha^\ast \pi (e_\rho) = u_\rho$, we find
\begin{align}
 \MV_{d} (\pi(S_{k+1}),\ldots, \pi(S_{k+d})) = d! \Vol_d (\Conv(u_1,u_2,\ldots, u_n)) \times \frac{1}{\sqrt{\det A^T A}}~.
\end{align}
The last factor is a normalization:  under the isomorphism $N_{\R} \to [\beta^\ast(N_1(X))]^\perp$, the unit parallelotope is mapped to a parallelotope with volume $\sqrt{\det A^T A}$.  Putting the factors together and using lemma~\ref{lem:detdet}~, we find that for generic parameter values the number of solutions to the QSCR is given by
\begin{align}
d! \Vol_d \Conv(u_1,u_2,\ldots, u_n)~.
\end{align}
To complete the proof of theorem~\ref{thm:QSCREuler}, we note that when $X$ is a smooth NEF Fano variety $\Conv(u_1,\ldots, u_n)$ can be subdivided into $\chi(X)$ $d$-dimensional simplices (one for each maximal cone in $\Sigma_X$), each with volume $1/d!$, and the assertion follows.

\subsection*{The discriminant locus}
Bernstein's theorem counts solutions to Laurent polynomial systems for generic values of the coefficients.   In addition there is a specific combinatorial condition that determines the non-generic locus~\cite{Bernstein:1975nr}.  When applied to the QSCR of a NEF Fano toric variety with $q_a \in \C^\ast$ this determines the A/2 discriminant in terms of the combinatoric structure.  As we saw in section~\ref{s:examples}, we do not need such machinery for simple examples.  Since it should be quite useful in analyzing more complicated examples, we summarize the relevant results of~\cite{Bernstein:1975nr} here.

Let $S = S_1 + \cdots S_n $ be the Newton polytope obtained by taking the Minkowski sum of the supports of the QSCR.\footnote{Note we will slightly abuse notation here and conflate the supports with their convex hulls.}  The polytope has the fan  $\Sigmat \subset \Nt_{\R}$, where $\Nt = \Mt^\ast$ is the dual lattice, and this fan defines a toric variety $\Xt$ that contains $T_n$ as its dense torus.  For any $w \in \Nt_{\R}$, Bernstein defines 
\begin{align}
m(w,S) = \min \left\{ \la w,p \ra~, p \in S \subset \Mt\right\}~,
\end{align}
and
\begin{align}
S^w = \left\{p \in S \subset \Mt ~~|~~ \la w,p \ra = m (w, S) \right\}~.
\end{align}
For any function $f = \sum_{p \in S} c_p z^p$, let $f^w = \sum_{p \in S^w} c_p z^p$.  For a system 
$F = (f_1,\ldots,f_n)$, where $f_i$ is supported on $S_i$, $F^w = (f_1^w,\ldots,f^w_n) = F^{w'}$ if  $w$ and $w'$ lie in the relative interior of the same cone $\sigmat \subset \Sigmat$.  The parameters $c_p$ are then on the discriminant locus (i.e. are non-generic) if and only if $F^w$ has roots in $T_n$ for some $w \neq 0$.

\section{Non-SUSY loci} \label{app:nonsusy}
In this appendix we provide the details for the non-SUSY loci for the models considered in the text.

\subsection{$X=\P^1\times\P^1$} \label{app:nonSUSYP1P1}
In this appendix we provide the details for the non-SUSY locus of the $\P^1\times\P^1$ model.  First we observe that whenever a column of the $E$ matrix vanishes, there are no SUSY vacua.  These loci correspond to the $\P^2_i$ defined by the simultaneous vanishing of $\zeta_{ij} = \zeta_{ik} =\zeta_{il} = 0$ for distinct $j,k,l \neq i$.  To find the remaining non-SUSY loci, we consider three cases:  $\zeta_{12} \neq 0$, $\zeta_{34} \neq 0$, and $\zeta_{12} = \zeta_{34} = 0$, and in each describe the possible degenerations of QSCR.\\

\noindent
When $\zeta_{12} \neq 0$, we use $\GL(2,\C)$ to bring $E$ to the following form:
\begin{align}
E =\begin{pmatrix} 1 & 0 & -\zeta_{23}/\zeta_{12}	&-\zeta_{24} /\zeta_{12} \\ 0 & 1 & \zeta_{13}/\zeta_{12} & \zeta_{14}/\zeta_{12} \end{pmatrix}~.
\end{align}
The QSCR can now be written as
\begin{align}
\sigma_1\sigma_2 & =q_1~,&
\zeta_{23}\zeta_{24} \sigma_1^4 - \Delta_{12} \sigma_1^2 + q_1^2 \zeta_{13} \zeta_{14} & = 0~,
\end{align}
where
\begin{align}
\Delta_{12} = \zeta_{12}^2 q_2 + (\zeta_{12}\zeta_{24} + \zeta_{14}\zeta_{23}) q_1~.
\end{align}
Generically, there are $4$ solutions, but there are more interesting possibilities.
\begin{enumerate}[(a)]
%%%%%% a
\item $\zeta_{23}\zeta_{24} = 0$, $\Delta_{12} = 0$, $\zeta_{13} \zeta_{14} \neq 0$.
\begin{enumerate}[(i)]
\item 
If $\zeta_{23} = 0$, then $\Delta_{12} = \zeta_{12} \Delta_+$; 
\item if $\zeta_{24} = 0$, then $\Delta_{12} = \zeta_{12} \Delta_-$, 
\end{enumerate}
where $\Delta_\pm = \zeta_{12} q_2 + \zeta_{34} q_1$ and we used $P=0$.
%%%%%% b
\item $\zeta_{13}\zeta_{14} = 0$, $\Delta_{12} = 0$, $\zeta_{23}\zeta_{24} \neq 0$.  
\begin{enumerate}[(i)]
\item
If $\zeta_{13} = 0$ then $\Delta_{12} = \zeta_{12} \Delta_-$;
\item if $\zeta_{14} = 0$, then $\Delta_{12} =\zeta_{12} \Delta_+$.
\end{enumerate}
%%%%%% c
\item $\zeta_{23}\zeta_{24} = 0$, $\zeta_{13}\zeta_{14} = 0$, $\Delta_{12} \neq 0$.  There are four possibilities here.  When we intersect these with $P=0$ we find the following.
\begin{enumerate}[(i)]
\item $\zeta_{23} = 0$, $\zeta_{13}=0$; $\implies$ $\zeta_{34} = 0$, $\Delta_{12} \neq 0$, i.e. it is the locus $\P^2_3$ in the $\zeta_{12} \neq 0$ patch;
\item $\zeta_{24} = 0$, $\zeta_{14} = 0$; $\implies$ $\zeta_{34} = 0$, $\Delta_{12} \neq 0$, i.e. it is the locus $\P^2_4$ in the $\zeta_{12} \neq 0$ patch;
\item $\zeta_{23} =0$, $\zeta_{14} = 0$; $\implies$ $\Delta_{12} = \zeta_{12} \Delta_+$;
\item $\zeta_{24} =0$, $\zeta_{13} = 0$; $\implies$ $\Delta_{12} = \zeta_{12} \Delta_-$;
\end{enumerate}
%%%%%% d
\item $\zeta_{23}\zeta_{24} = 0$, $\zeta_{13}\zeta_{14} = 0$, $\Delta_{12} = 0$.  There are just two non-trivial components here:
\begin{align}
\{\zeta_{23} = 0,~\zeta_{14} = 0,~ \Delta_+ = 0\} \cup
\{\zeta_{24} = 0,~\zeta_{13} = 0,~ \Delta_- = 0 \}~.
\end{align}
\end{enumerate}
Case (d) leads to a continuum of $\sigma$ vacua, i.e. it contributes to the locus $\cA_+$; cases (a)--(c) have no SUSY vacua. \\

\noindent When $\zeta_{34} \neq 0$, we obtain a very similar story, where after a change of basis the QSCR become
\begin{align}
\sigma_1\sigma_2 & =q_2~,&
\zeta_{14}\zeta_{24} \sigma_1^4 - \Delta_{34} \sigma_1^2 + q_2^2 \zeta_{13} \zeta_{23} & = 0~,
\end{align}
with $\Delta_{34} = \zeta_{34}^2 q_1 + (\zeta_{13}\zeta_{24} + \zeta_{14}\zeta_{23}) q_2$.
The only new loci we obtain are $\P^2_1$ and $\P^2_2$ restricted to $\zeta_{34} \neq 0$. \\

\noindent Finally, when $\zeta_{12} = \zeta_{34} = 0$, then $P = 0$ reduces to $\zeta_{14} \zeta_{23} = \zeta_{13} \zeta_{24}$, so that points with $\zeta_{13} = 0$ belong to either $\P^2_1$ or $\P^2_3$.  If $\zeta_{13} \neq 0$, then the QSCR can be brought to the form
\begin{align}
\zeta_{23} \sigma_1^2 &= \zeta_{13} q_1~,&
\zeta_{14} \sigma_2^2 & = \zeta_{13} q_2~.
\end{align}
Non-SUSY configurations can therefore only be obtained if $\zeta_{23} = 0$ or if $\zeta_{14} =0$, and that means the point belongs to either $\P^2_2$ or $\P^2_4$.\\

\noindent  Reorganizing these points a bit leads to the result in the text.

\subsection{$X=dP_1$} \label{app:nonSUSYdP1}
 As for $\P^1\times\P^1$, whenever a column of the $E$ matrix vanishes, there are no SUSY vacua.  Thus $\cA_{\nSUSY}$ contains $\P^2_i$ defined by the simultaneous vanishing of $\zeta_{ij} = \zeta_{ik} =\zeta_{il} = 0$ for distinct $j,k,l \neq i$.  We will show that there are no additional points in $\cA_{\nSUSY}$.

Consider first the patch where $\zeta_{34} \neq 0$, where by a choice of basis
\begin{align}
E = \begin{pmatrix} \zeta_{14} /\zeta_{34} & \zeta_{24}/\zeta_{34} & 1 & 0 \\ -\zeta_{13}/\zeta_{34} & -\zeta_{23} /\zeta_{34} & 0 & 1
\end{pmatrix}~,
\end{align}
which with a small manipulation leads to the QSCR
\begin{align}
\sigma_1\sigma_2 & = q_2~,&
\zeta_{14}\zeta_{24} q_2 \sigma_1 & = \zeta_{34}^2 q_1 + (\zeta_{14}\zeta_{23} + \zeta_{13}\zeta_{24}) q_2 \sigma_2 -\zeta_{13}\zeta_{23} \sigma_2^3~.
\end{align}
For finite solutions $\sigma_a \in \C^\ast$, so that we may equivalently consider
\begin{align}
\sigma_1\sigma_2 & = q_2~,&
\zeta_{14}\zeta_{24} q_2^2 & = \zeta_{34}^2 q_1 \sigma_2+ (\zeta_{14}\zeta_{23} + \zeta_{13}\zeta_{24}) q_2 \sigma^2_2 -\zeta_{13}\zeta_{23} \sigma_2^4~,
\end{align}
and since $\zeta_{34}^2 q_1 \neq 0$, this has a non-zero solution unless
\begin{align}
\zeta_{14} \zeta_{24} &= 0~,&
\zeta_{13}\zeta_{23} &= 0~,&
\zeta_{14}\zeta_{23} + \zeta_{13}\zeta_{24} & = 0~.
\end{align}
Intersecting this with $P = \zeta_{12}\zeta_{34} +\zeta_{14}\zeta_{23} -\zeta_{13}\zeta_{24} = 0$,  it is easy to check that every solution belongs to some $\P^2_i$.\\

Next, consider the locus where $\zeta_{34} = 0$, but $\zeta_{14} \neq 0$.  Now
\begin{align}
E = \begin{pmatrix} 1 & \zeta_{24}/\zeta_{14} & 0 & 0 \\ 0 & \zeta_{12}/\zeta_{14} & \zeta_{13}/\zeta_{14} & 1 
\end{pmatrix}~,
\end{align}
which leads to the QSCR
\begin{align}
\zeta_{24} \sigma_1^2\sigma_2 +\zeta_{12} \sigma_1 \sigma_2^2 & = \zeta_{14}^2 q_1~,&
\zeta_{13} \sigma_2^2 & = \zeta_{14} q_2~.
\end{align}
A solution exists unless $\zeta_{13} = 0$ or $\zeta_{24} = \zeta_{12} = 0$.  In either case this belongs to a $\P^2_i$.\\

Finally, we consider $\zeta_{34} = 0$, $\zeta_{14} = 0$, and $\zeta_{24} \neq 0$.\footnote{If we instead set $\zeta_{24} = 0$, we obtain $\P^2_4$.}  This leads to
\begin{align}
E = \begin{pmatrix} 0 & 1 & 0 & 0 \\  -\zeta_{12}/\zeta_{24} & 0 &-\zeta_{13}/\zeta_{24} & 1 
\end{pmatrix}~,
\end{align}
and QSCR
\begin{align}
-\zeta_{12} \sigma_1\sigma_2^2 & = \zeta_{24} q_1~,&
-\zeta_{13} \sigma_2^2 & = \zeta_{24} q_2~.
\end{align}
There are solutions unless $\zeta_{12} = 0$ or $\zeta_{13} = 0$, and in either case this is a point in $\P^2_i$.  Hence, $\cA_{\nSUSY}$ is exactly the union of the $\P^2_i$.

\subsection{$X=\F_2$}\label{app:nonSUSYF2}
For our final example, we will not work out the general non-SUSY locus in $\Gr(2,4)\times(\C^\ast)^2$, but merely its intersection with the complement of the exceptional set $F_C$.  This is sufficient for describing the non-SUSY locus for the stability condition defined by the cone $C$ in figure~\ref{fig:stabforF2}.  For convenience we reproduce the exceptional set here:
\begin{align}
F_C = \{\zeta_{13} =0\}\cup \{\zeta_{14} =0,\zeta_{34} =0\} \cup \{\zeta_{12}=0,\zeta_{23}=0,\zeta_{24} = 0\} \subset \{D=0\}
\end{align}

Suppose, for starters, that in addition to $\zeta_{13} \neq 0$ we also have $\zeta_{12} \neq 0$.  In this case the QSCR
take the form
%\begin{align}
%\sigma_1\sigma_2 & = q_1~,&
%(-\zeta_{23}\sigma_1+\zeta_{13}\sigma_2)(-\zeta_{24}\sigma_1+\zeta_{14}\sigma_2) & = \zeta_{12}^2q_2\sigma_1^2~,
%\end{align}
\begin{align}
\sigma_1\sigma_2 & = q_1~,&
\Delta_1 \sigma_1^4 - \Delta_2 q_1 \sigma_1^2 +\zeta_{13}\zeta_{14} q_1^2 & = 0~,
\end{align}
where
\begin{align}
\Delta_1 & = \zeta_{23}\zeta_{24} -\zeta_{12}^2 q_2~,&
\Delta_2 &  = \zeta_{13}\zeta_{24} + \zeta_{14}\zeta_{23}~.
\end{align}
The second equation is trivial, which would lead to points in $\cA_+$, if and only if
\begin{align}
\label{eq:inf}
\Delta_1 = \Delta_2 = \zeta_{13}\zeta_{14} = 0~.
\end{align}
Once we use $P=0$, it is easy to see that $\cA_+ \cap \{\zeta_{12} \neq 0\}$ is contained in $F_C$.

The other possibility is that the second equation has no solutions, which takes place on the following loci.
\begin{enumerate}[(i)]
\item $\Delta_1 = \Delta_2 = 0$.  Upon intersection with $P=0$, this is equivalent (in this patch) to
\begin{align}
\Delta'_1 =\zeta_{23}\zeta_{34}-2q_2 \zeta_{12}\zeta_{13} & = 0~,& \Delta_2 & = 0~,& P&=0.
\end{align}
\item $\Delta_1 = \zeta_{14} = 0$;
\item $\Delta_2 = \zeta_{14} = 0$.  The intersection of this locus with $P=0$ is contained in $F_C$.
\end{enumerate}
Next, we consider $\zeta_{12} = 0$ and $\zeta_{13} \neq 0$.  This leads to QSCR
\begin{align}
\zeta_{23} \sigma_1^2 & = \zeta_{13} q_1~,&
\zeta_{14}\sigma_2^2 -\zeta_{34}\sigma_1\sigma_2 & = \zeta_{13} q_2 \sigma_1^2~.
\end{align}
If $\zeta_{23} =0$ then $P=0$ implies $\zeta_{24} =0$, which is contained in $F_C$.  If $\zeta_{23}\neq0$ then solutions fail to exist if and only if $\zeta_{14} =\zeta_{34} =0$, but this again in $F_C$.

Putting this together, we see that the non-SUSY locus in the complement of $F_C$ is $\cA = \{P=0\} \cap \cA'$, with
\begin{align}
\cA' = \{\Delta'_1 = 0,~\Delta_2 = 0\} \cup
          \{\Delta'_1 = 0,~\zeta_{14}= 0\}~.
\end{align}

%Finally, we consider $\zeta_{12} = \zeta_{13} = 0$ and $\zeta_{14} \neq 0$, where the QSCR are
%\begin{align}
%\zeta_{24}\sigma_1^2 & = \zeta_{14} q_1~,&
%\zeta_{34} \sigma_1\sigma_2 = q_2 \zeta_{14} \sigma_1^2~.
%\end{align}
%A solution fails to exist only on an intersection with a $\P^2_i$.  Again, $\cA_+$ does not intersect this locus.

%\bibliographystyle{./utphys}
%\bibliography{/Users/lmel/BIB/bigref}

\begin{thebibliography}{10}

\bibitem{Melnikov:2012hk}
I.~Melnikov, S.~Sethi, and E.~Sharpe, ``{Recent Developments in (0,2) Mirror
  Symmetry},'' \href{http://dx.doi.org/10.3842/SIGMA.2012.068}{{\em SIGMA} {\bf
  8} (2012)  068},
\href{http://arxiv.org/abs/1209.1134}{{\tt arXiv:1209.1134 [hep-th]}}.
%%CITATION = ARXIV:1209.1134;%%.

\bibitem{Witten:1993yc}
E.~Witten, ``{Phases of N = 2 theories in two dimensions},'' {\em Nucl. Phys.}
  {\bf B403} (1993)  159--222,
\href{http://arxiv.org/abs/hep-th/9301042}{{\tt arXiv:hep-th/9301042}}.
%%CITATION = HEP-TH/9301042;%%.

\bibitem{Distler:1993mk}
J.~Distler and S.~Kachru, ``(0,2) {L}andau-{G}inzburg theory,'' {\em Nucl.
  Phys.} {\bf B413} (1994)  213--243,
\href{http://arxiv.org/abs/hep-th/9309110}{{\tt hep-th/9309110}}.
%%CITATION = HEP-TH/9309110;%%.

\bibitem{Distler:1996tj}
J.~Distler, B.~R. Greene, and D.~R. Morrison, ``{Resolving singularities in
  (0,2) models},'' {\em Nucl. Phys.} {\bf B481} (1996)  289--312,
\href{http://arxiv.org/abs/hep-th/9605222}{{\tt arXiv:hep-th/9605222}}.
%%CITATION = HEP-TH/9605222;%%.

\bibitem{Chiang:1997kt}
T.-M. Chiang, J.~Distler, and B.~R. Greene, ``{Some features of (0,2) moduli
  space},'' {\em Nucl. Phys.} {\bf B496} (1997)  590--616,
\href{http://arxiv.org/abs/hep-th/9702030}{{\tt arXiv:hep-th/9702030}}.
%%CITATION = HEP-TH/9702030;%%.

\bibitem{Silverstein:1995re}
E.~Silverstein and E.~Witten, ``{Criteria for conformal invariance of (0,2)
  models},'' {\em Nucl. Phys.} {\bf B444} (1995)  161--190,
\href{http://arxiv.org/abs/hep-th/9503212}{{\tt arXiv:hep-th/9503212}}.
%%CITATION = HEP-TH/9503212;%%.

\bibitem{Berglund:1995yu}
P.~Berglund, P.~Candelas, X.~de~la Ossa, E.~Derrick, J.~Distler, {\em et al.},
  ``{On the instanton contributions to the masses and couplings of E(6)
  singlets},'' {\em Nucl.Phys.} {\bf B454} (1995)  127--163,
  \href{http://arxiv.org/abs/hep-th/9505164}{{\tt arXiv:hep-th/9505164
  [hep-th]}}.

\bibitem{Basu:2003bq}
A.~Basu and S.~Sethi, ``World-sheet stability of (0,2) linear sigma models,''
  {\em Phys. Rev.} {\bf D68} (2003)  025003,
\href{http://arxiv.org/abs/hep-th/0303066}{{\tt hep-th/0303066}}.
%%CITATION = HEP-TH/0303066;%%.

\bibitem{Beasley:2003fx}
C.~Beasley and E.~Witten, ``{Residues and world-sheet instantons},'' {\em JHEP}
  {\bf 10} (2003)  065,
\href{http://arxiv.org/abs/hep-th/0304115}{{\tt arXiv:hep-th/0304115}}.
%%CITATION = HEP-TH/0304115;%%.

\bibitem{McOrist:2008ji}
J.~McOrist and I.~V. Melnikov, ``{Summing the instantons in half-twisted linear
  sigma models},'' {\em JHEP} {\bf 02} (2009)  026,
\href{http://arxiv.org/abs/0810.0012}{{\tt arXiv:0810.0012 [hep-th]}}.
%%CITATION = 0810.0012;%%.

\bibitem{MR2931867}
T.~de~Fernex and C.~D. Hacon, ``Rigidity properties of {F}ano varieties,'' in
  {\em Current developments in algebraic geometry}, vol.~59 of {\em Math. Sci.
  Res. Inst. Publ.}, pp.~113--127.
\newblock Cambridge Univ. Press, Cambridge, 2012.

\bibitem{Kreuzer:2010ph}
M.~Kreuzer, J.~McOrist, I.~V. Melnikov, and M.~Plesser, ``{(0,2) deformations
  of linear sigma models},''
  \href{http://dx.doi.org/10.1007/JHEP07(2011)044}{{\em JHEP} {\bf 1107} (2011)
   044}, \href{http://arxiv.org/abs/1001.2104}{{\tt arXiv:1001.2104 [hep-th]}}.

\bibitem{Aspinwall:1993rj}
P.~S. Aspinwall, B.~R. Greene, and D.~R. Morrison, ``{The monomial divisor
  mirror map},'' {\em Internat. Math. Res. Notices} (1993) no.~12, 319--337,
\href{http://arxiv.org/abs/alg-geom/9309007}{{\tt arXiv:alg-geom/9309007}}.
%%CITATION = ALG-GEOM/9309007;%%.

\bibitem{Cox:2000vi}
D.~A. Cox and S.~Katz, ``{Mirror symmetry and algebraic geometry},''.
  Providence, USA: AMS (2000) 469 p.

\bibitem{DonagiLuIVM2}
R.~Donagi, Z.~Lu, and I.~Melnikov To appear.

\bibitem{Tong:2008qd}
D.~Tong, ``{Quantum Vortex Strings: A Review},''
  \href{http://dx.doi.org/10.1016/j.aop.2008.10.005}{{\em Annals Phys.} {\bf
  324} (2009)  30--52},
\href{http://arxiv.org/abs/0809.5060}{{\tt arXiv:0809.5060 [hep-th]}}.
%%CITATION = ARXIV:0809.5060;%%.

\bibitem{Gadde:2013lxa}
A.~Gadde, S.~Gukov, and P.~Putrov, ``{(0,2) Trialities},''
\href{http://arxiv.org/abs/1310.0818}{{\tt arXiv:1310.0818 [hep-th]}}.
%%CITATION = ARXIV:1310.0818;%%.

\bibitem{Dumitrescu:2011iu}
T.~T. Dumitrescu and N.~Seiberg, ``{Supercurrents and Brane Currents in Diverse
  Dimensions},'' \href{http://dx.doi.org/10.1007/JHEP07(2011)095}{{\em JHEP}
  {\bf 1107} (2011)  095},
\href{http://arxiv.org/abs/1106.0031}{{\tt arXiv:1106.0031 [hep-th]}}.
%%CITATION = ARXIV:1106.0031;%%.

\bibitem{Morrison:1994fr}
D.~R. Morrison and M.~Ronen~Plesser, ``{Summing the instantons: quantum
  cohomology and mirror symmetry in toric varieties},'' {\em Nucl. Phys.} {\bf
  B440} (1995)  279--354,
\href{http://arxiv.org/abs/hep-th/9412236}{{\tt arXiv:hep-th/9412236}}.
%%CITATION = HEP-TH/9412236;%%.

\bibitem{McOrist:2007kp}
J.~McOrist and I.~V. Melnikov, ``{Half-twisted correlators from the Coulomb
  branch},'' {\em JHEP} {\bf 04} (2008)  071,
\href{http://arxiv.org/abs/0712.3272}{{\tt arXiv:0712.3272 [hep-th]}}.
%%CITATION = 0712.3272;%%.

\bibitem{Gaiotto:2013sma}
D.~Gaiotto, S.~Gukov, and N.~Seiberg, ``{Surface Defects and Resolvents},''
  \href{http://dx.doi.org/10.1007/JHEP09(2013)070}{{\em JHEP} {\bf 1309} (2013)
   070},
\href{http://arxiv.org/abs/1307.2578}{{\tt arXiv:1307.2578 [hep-th]}}.
%%CITATION = ARXIV:1307.2578;%%.

\bibitem{Adams:2003zy}
A.~Adams, A.~Basu, and S.~Sethi, ``(0,2) duality,'' {\em Adv. Theor. Math.
  Phys.} {\bf 7} (2004)  865--950,
\href{http://arxiv.org/abs/hep-th/0309226}{{\tt hep-th/0309226}}.
%%CITATION = HEP-TH/0309226;%%.

\bibitem{Katz:2004nn}
S.~H. Katz and E.~Sharpe, ``Notes on certain (0,2) correlation functions,''
  {\em Commun. Math. Phys.} {\bf 262} (2006)  611--644,
\href{http://arxiv.org/abs/hep-th/0406226}{{\tt hep-th/0406226}}.
%%CITATION = HEP-TH/0406226;%%.

\bibitem{Adams:2005tc}
A.~Adams, J.~Distler, and M.~Ernebjerg, ``Topological heterotic rings,'' {\em
  Adv. Theor. Math. Phys.} {\bf 10} (2006)  657--682,
\href{http://arxiv.org/abs/hep-th/0506263}{{\tt hep-th/0506263}}.
%%CITATION = HEP-TH/0506263;%%.

\bibitem{Donagi:2011uz}
R.~Donagi, J.~Guffin, S.~Katz, and E.~Sharpe, ``{A mathematical theory of
  quantum sheaf cohomology},'' \href{http://arxiv.org/abs/1110.3751}{{\tt
  arXiv:1110.3751 [math.AG]}}.

\bibitem{Donagi:2011va}
R.~Donagi, J.~Guffin, S.~Katz, and E.~Sharpe, ``{Physical aspects of quantum
  sheaf cohomology for deformations of tangent bundles of toric varieties},''
  \href{http://arxiv.org/abs/1110.3752}{{\tt arXiv:1110.3752 [hep-th]}}.

\bibitem{Klyachko:1989eb}
A.~A. Klyachko, ``Equivariant bundles over toric varieties,'' {\em Izv. Akad.
  Nauk SSSR Ser. Mat.} {\bf 53} (1989) no.~5, 1001--1039, 1135.

\bibitem{Knutson:1997yt}
A.~Knutson and E.~R. Sharpe, ``{Sheaves on toric varieties for physics},'' {\em
  Adv.Theor.Math.Phys.} {\bf 2} (1998)  865--948,
\href{http://arxiv.org/abs/hep-th/9711036}{{\tt arXiv:hep-th/9711036
  [hep-th]}}.
%%CITATION = HEP-TH/9711036;%%.

\bibitem{Payne:2008tb}
S.~Payne, ``Moduli of toric vector bundles,'' {\em Compos. Math.} {\bf 144}
  (2008) no.~5, 1199--1213.

\bibitem{Cox:2011tv}
D.~Cox, J.~Little, and H.~Schenck, {\em Toric varieties}, vol.~124 of {\em
  Graduate Studies in Mathematics}.
\newblock AMS, 2011.

\bibitem{MR1299003}
D.~A. Cox, ``The homogeneous coordinate ring of a toric variety (+erratum,
  2014),'' {\em J. Algebraic Geom.} {\bf 4} (1995) no.~1, 17--50,
  \href{http://arxiv.org/abs/alg-geom/9210008v3}{{\tt
  arXiv:alg-geom/9210008v3}}.

\bibitem{Melnikov:2006kb}
I.~V. Melnikov and M.~R. Plesser, ``A-model correlators from the {C}oulomb
  branch,'' {\em JHEP} {\bf 02} (2006)  044,
\href{http://arxiv.org/abs/hep-th/0507187}{{\tt hep-th/0507187}}.
%%CITATION = JHEPA,0602,044;%%.

\bibitem{MR1265307}
V.~V. Batyrev, ``Quantum cohomology rings of toric manifolds,'' {\em
  Ast\'erisque} (1993) no.~218, 9--34,
  \href{http://arxiv.org/abs/alg-geom/9310004}{{\tt arXiv:alg-geom/9310004}}.

\bibitem{Benini:2013nda}
F.~Benini, R.~Eager, K.~Hori, and Y.~Tachikawa, ``{Elliptic genera of
  two-dimensional N=2 gauge theories with rank-one gauge groups},''
\href{http://arxiv.org/abs/1305.0533}{{\tt arXiv:1305.0533 [hep-th]}}.
%%CITATION = ARXIV:1305.0533;%%.

\bibitem{Gadde:2013dda}
A.~Gadde and S.~Gukov, ``{2d Index and Surface operators},''
\href{http://arxiv.org/abs/1305.0266}{{\tt arXiv:1305.0266 [hep-th]}}.
%%CITATION = ARXIV:1305.0266;%%.

\bibitem{Benini:2013xpa}
F.~Benini, R.~Eager, K.~Hori, and Y.~Tachikawa, ``{Elliptic genera of 2d N=2
  gauge theories},''
\href{http://arxiv.org/abs/1308.4896}{{\tt arXiv:1308.4896 [hep-th]}}.
%%CITATION = ARXIV:1308.4896;%%.

\bibitem{Vafa:1990mu}
C.~Vafa, ``{Topological Landau-Ginzburg models},''
{\em Mod. Phys. Lett.} {\bf A6} (1991)  337--346.
%%CITATION = MPLAE,A6,337;%%.

\bibitem{Mumford:GIT}
D.~Mumford, J.~Fogarty, and F.~Kirwan, {\em Geometric invariant theory},
  vol.~34 of {\em Ergebnisse der Mathematik und ihrer Grenzgebiete (2) [Results
  in Mathematics and Related Areas (2)]}.
\newblock Springer-Verlag, Berlin, third~ed., 1994.

\bibitem{Seiberg:1994bp}
N.~Seiberg, ``{The Power of holomorphy: Exact results in 4-D SUSY field
  theories},''
\href{http://arxiv.org/abs/hep-th/9408013}{{\tt arXiv:hep-th/9408013
  [hep-th]}}.
%%CITATION = HEP-TH/9408013;%%.

\bibitem{Intriligator:1994jr}
K.~A. Intriligator, R.~Leigh, and N.~Seiberg, ``{Exact superpotentials in
  four-dimensions},'' \href{http://dx.doi.org/10.1103/PhysRevD.50.1092}{{\em
  Phys.Rev.} {\bf D50} (1994)  1092--1104},
\href{http://arxiv.org/abs/hep-th/9403198}{{\tt arXiv:hep-th/9403198
  [hep-th]}}.
%%CITATION = HEP-TH/9403198;%%.

\bibitem{Argyres:book}
P.~C. Argyres, {\em An introduction to global supersymmetry}.
\newblock DIY, 2000.

\bibitem{Bernstein:1975nr}
D.~Bernstein, ``The number of roots of a system of equations,'' {\em Functional
  Analysis and Its Applications} {\bf 9} (1975) no.~3, 183--185.

\bibitem{Fulton:1993tv}
W.~Fulton, {\em Introduction to toric varieties}.
\newblock Princeton University Press, 1993.

\bibitem{Sturmfels:2002sp}
B.~Sturmfels, {\em Solving systems of polynomial equations}.
\newblock No.~97 in Regional Conference Series in Mathematics. American
  Mathematical Society, 2002.

\bibitem{Cox:1998ua}
D.~Cox, J.~Little, and D.~{O'Shea}, {\em Using Algebraic Geometry}.
\newblock Graduate Texts in Mathematics. Springer, 1998.

\bibitem{Steffens:2009mixed}
R.~Steffens, {\em Mixed volumes, mixed Ehrhart theory and applications to
  tropical geometry and linkage configurations}.
\newblock PhD thesis, {Goethe Universit\"at} Frankfurt am Main, 2009.

\end{thebibliography}

\providecommand{\href}[2]{#2}\begingroup\raggedright\endgroup

\end{document}